\documentclass[oneside,english]{amsart}
\usepackage[T1]{fontenc}
\usepackage[utf8]{inputenc}
\usepackage{color}
\usepackage{babel}
\usepackage{float}
\usepackage{booktabs}
\usepackage{url}
\usepackage{bm}
\usepackage{multirow}
\usepackage{amsbsy}
\usepackage{amstext}
\usepackage{amsthm}
\usepackage{amssymb}
\usepackage{graphicx}
\usepackage[pdfusetitle,
 bookmarks=true,bookmarksnumbered=false,bookmarksopen=false,
 breaklinks=false,pdfborder={0 0 1},backref=false,colorlinks=true]
 {hyperref}

\makeatletter

\providecommand{\tabularnewline}{\\}
\floatstyle{ruled}
\newfloat{algorithm}{tbp}{loa}
\providecommand{\algorithmname}{Algorithm}
\floatname{algorithm}{\protect\algorithmname}

\numberwithin{equation}{section}
\numberwithin{figure}{section}
\theoremstyle{plain}
\newtheorem{thm}{\protect\theoremname}[section]
\theoremstyle{definition}
\newtheorem{defn}[thm]{\protect\definitionname}
\theoremstyle{plain}
\newtheorem{lem}[thm]{\protect\lemmaname}
\theoremstyle{plain}
\newtheorem{cor}[thm]{\protect\corollaryname}
\theoremstyle{plain}
\newtheorem{prop}[thm]{\protect\propositionname}

\usepackage{amsthm}

\usepackage{algorithm}
\usepackage{algpseudocode}

\ifdefined\showcaptionsetup
 \PassOptionsToPackage{caption=false}{subfig}
\fi
\usepackage{subfig}
\makeatother

\providecommand{\corollaryname}{Corollary}
\providecommand{\definitionname}{Definition}
\providecommand{\lemmaname}{Lemma}
\providecommand{\propositionname}{Proposition}
\providecommand{\theoremname}{Theorem}

\begin{document}

\numberwithin{equation}{section}
\newtheorem{theorem}{Theorem}[section]

\title{Legendre Polynomials and Their Use for Karhunen--Loève Expansion}

\author{Michal Béreš}
\address{Institute of Geonics of the CAS, Ostrava, Czech Republic}
\address{Department of Applied Mathematics, Faculty of Electrical Engineering and Computer Science,
V\v{S}B – Technical University of Ostrava, Ostrava, Czech Republic}
\email{michal.beres@vsb.cz}

\date{\today}

\dedicatory{Dedicated to the memory of Prof. Radim Blaheta --- an inspiring teacher, generous mentor, and truly kind person.}

\begin{abstract}
This paper makes two main contributions. First, we present a pedagogical review of the derivation of the three-term recurrence relation for Legendre polynomials, without relying on the classical Legendre differential equation, Rodrigues' formula, or generating functions. This exposition is designed to be accessible to undergraduate students.

Second, we develop a computational framework for Karhunen--Loève expansions of isotropic Gaussian random fields on hyper-rectangular domains. The framework leverages Legendre polynomials and their associated Gaussian quadrature, and it remains efficient even in higher spatial dimensions.

A covariance kernel is first approximated by a non-negative mixture of squared-exponentials, obtained via a Newton-optimized fit with a theoretically informed initialization. The resulting separable kernel enables a Legendre--Galerkin discretization in the form of a Kronecker product over single dimensions, with submatrices that exhibit even/odd parity structure. For assembly, we introduce a Duffy-type transformation followed by quadrature. These structural properties significantly reduce both memory usage and arithmetic cost compared to naive approaches. All algorithms and numerical experiments are provided in an open-source repository that reproduces every figure and table in this work.
\end{abstract}

\keywords{Karhunen--Loève expansion, Legendre polynomials, squared-exponential approximation}
\subjclass[2020]{65C05, 86-08, 82-08, 65C60, 60-08}

\thanks{This work was supported by the European Union under Grant Agreement no.~101166718 (EURAD-2 project).\\
This article was co-funded by the European Union under the REFRESH -- Research Excellence For REgion Sustainability and High-tech Industries project number \url{CZ.10.03.01/00/22_003/0000048} via the Operational Programme Just Transition.}

\maketitle

\section{Introduction}

\label{sec:introduction}

\subsection{Motivation}

The rapid growth of computational power has made it feasible to solve
large systems of partial differential equations (PDEs) with uncertain
input data. A prominent example arises in geoscience, where the hydraulic
conductivity of an aquifer is often sparsely measured and must therefore
be described statistically \cite{Freeze1975,Christakos1992}. Such
situations call for realistic \emph{Gaussian random fields} (GRFs),
whose samples or low-rank representations can be incorporated into
a downstream PDE solver. Among the available dimension-reduction tools,
the Karhunen--Loève (KL) expansion is preferred because it minimizes
the mean-square truncation error for any fixed number of terms \cite{Loeve1978,GhanemSpanos1991}.
Unfortunately, a straightforward Galerkin discretization of the KL
integral operator yields dense matrices, whose assembly and eigendecomposition
become computationally prohibitive in more than two spatial dimensions. 

\subsection{State of the Art}

Polynomial--Galerkin solvers for the Fredholm eigenproblem underlying
the Karhunen--Loève (KL) expansion have steadily advanced. Early
convergence analyses for low-order bases \cite{Huang2001} and benchmark
comparisons of Nyström, collocation, and finite element methods \cite{Betz2014}
identified dense matrix assembly as the primary bottleneck. Exponential
accuracy with Legendre spectral elements \cite{Oliveira2014}, multilevel
correction strategies \cite{XieZhou2015}, sparsity-oriented Chebyshev
and Haar variants \cite{LiuZhang2017,Azevedo2017}, and isogeometric
formulations enabling matrix-free tensor contractions \cite{Rahman2018,Mika2021}
have all reduced computational costs in specific settings. A discontinuous
Legendre approach further exploits block-diagonal mass matrices \cite{Basmaji2022}.
Yet, to the best of our knowledge, no existing method simultaneously
leverages 
\begin{enumerate}
\item[(i)] the separability introduced by non-negative squared-exponential fits
of isotropic covariances, 
\item[(ii)] the Kronecker structure arising from separability, hyper-rectangular
domains, and tensor-product bases, 
\item[(iii)] and the even/odd parity symmetry of Legendre modes, which leads to
block-diagonal matrices. 
\end{enumerate}
The framework developed in this paper unifies these three ingredients:
a squared-exponential kernel approximation feeds a tensor-product
Legendre--Galerkin discretisation; parity splits the spectrum into
independent blocks; and each block is stored and applied as a Kronecker
product. This synergy reduces memory and arithmetic requirements by
orders of magnitude while preserving the KL expansion’s optimal mean-square
truncation error. 

\subsection{Contributions}

This paper addresses the bottleneck caused by dense matrices arising
from Galerkin discretisation and aims to keep the mathematics accessible
to a broad audience. It makes two contributions: 
\begin{enumerate}
\item \emph{Elementary derivation of Legendre polynomials.} In Section~\ref{sec:derivation},
we re-derive the three-term recurrence relation for Legendre orthogonal
polynomials without invoking the Legendre differential equation, Rodrigues'
formula, or generating functions. The exposition relies solely on
elementary linear algebra and mathematical analysis, making it suitable
for undergraduates. 
\item \emph{Efficient KL framework for isotropic GRFs.} Sections~\ref{sec:Karhunen-Loeve-expansion-of}--\ref{sec:Numerical-approximation-of}
develop a pipeline for KL expansions on $D$-dimensional hyper-rectangles.
The key ingredients are: 
\begin{itemize}
\item a covariance kernel fitted by a \textbf{non-negative} mixture of squared-exponentials
(Section~\ref{sec:gauss_sum}), obtained via a Newton solver with
an initial guess guided by a Gaussian quadrature interpretation; 
\item a Legendre--Galerkin discretisation that exploits the separability
of the fitted kernel to form a Kronecker sum of one-dimensional blocks
(Section~\ref{sec:galerkin}); 
\item a tensor-structured assembly whose submatrices decompose into even/odd
parity blocks and are evaluated stably using a Duffy-type quadrature
(Sections~\ref{subsec:tensor-assembly}--\ref{subsec:duffy-quadrature}). 
\end{itemize}
Together, these components significantly reduce both the memory footprint
and the computational cost of computing the KL decomposition. 
\end{enumerate}
This paper builds upon and extends our earlier studies \cite{Beres2018,Beres2023},
which demonstrated the suitability of a tensor-product Legendre--Galerkin
discretisation for the KL decomposition of higher-dimensional fields.
Compared to previous work, we enhance all aspects of the computational
pipeline and provide theoretical references to support the foundations
of the numerical approximation. 

\subsection{Reproducibility}

All code implementing the functionality described in this paper, as
well as all scripts required to reproduce the figures, tables, and
numerical experiments, are available in a public repository at \url{https://github.com/Beremi/KL_decomposition}.

\section{Derivation of Legendre Polynomials and Corresponding Quadrature Rule}

\label{sec:derivation}

This section introduces the basic theory of orthogonal polynomials
and re-derives the classical \emph{three-term recurrence relation}
for Legendre polynomials. The approach taken---particularly in deriving
polynomial norms---prioritizes accessibility, making it potentially
suitable for undergraduate students by relying primarily on recurrence
relations, rather than more advanced techniques such as the Legendre
differential equation, Rodrigues' formula, or generating functions
(see, e.g.,~\cite{Szego1975}).

We begin with the three-term recurrence for monic polynomials, and
then use it to derive the recurrence for normalized polynomials in
Subsection~\ref{subsec:Three-term-recurrence-relation-normalized}
and for classical Legendre polynomials in Subsection~\ref{subsec:Three-term-recurrence-relation-classical}.

In Subsection~\ref{subsec:GW_quadrature}, we briefly recall the
connection between the Gauss quadrature rule and orthogonal polynomials
in the context of efficient quadrature computation, as introduced
in~\cite{GolubWelsch1969}. The recurrence relation---specifically
for normalized polynomials---as well as the quadrature rule, will
be used later for approximating the KL decomposition. 

\subsection{Preliminaries}
\begin{defn}[Positive measure \cite{Szego1975,Chihara1978}]
\label{def:pos_def_measure} Let $d\lambda(t)$ be a measure on an
interval $[a,b]$ (where $a\in\mathbb{R}\cup\{-\infty\}$, $b\in\mathbb{R}\cup\{\infty\}$,
and $a<b$), such that the bilinear form 
\[
\langle f,g\rangle_{d\lambda}=\int_{a}^{b}f(t)g(t)\,d\lambda(t)
\]
defines an inner product on the space of real polynomials. Such a
measure is called \emph{positive}. 
\end{defn}

For a positive measure, the moments $\mu_{k}=\int_{a}^{b}t^{k}\,d\lambda(t)$
exist for all $k\ge0$. 
\begin{defn}[Monic polynomials]
\label{def:monic} A polynomial $p(t)$ of degree $k$ is called
\emph{monic} if the coefficient of its highest-degree term ($t^{k}$)
is equal to $1$. That is, 
\[
p(t)=t^{k}+c_{k-1}t^{k-1}+\dots+c_{1}t+c_{0}.
\]
We will denote a monic polynomial of degree $k$ by $P_{k}^{\text{M}}$.
\end{defn}

\subsection{Three-Term Recurrence of Legendre Polynomials}

A fundamental property of orthogonal polynomials generated by a positive
measure and the corresponding integral inner product is that they
satisfy a three-term recurrence relation. 
\begin{thm}[\cite{Szego1975,Chihara1978,Gautschi2004}]
\label{thm:monic_recurrence} Let $\pi_{k}(t)$, $k=0,1,2,\ldots$,
be the sequence of monic orthogonal polynomials with respect to the
positive measure $d\lambda$ on $[a,b]$. Then, they satisfy the recurrence
relation: 
\begin{gather*}
\pi_{k+1}(t)=(t-\alpha_{k})\pi_{k}(t)-\beta_{k}\pi_{k-1}(t),\quad k=0,1,2,\ldots,\\
\pi_{-1}(t)\equiv0,\quad\pi_{0}(t)\equiv1,
\end{gather*}
where the coefficients are given by 
\begin{align*}
\alpha_{k} & =\frac{\langle t\pi_{k},\pi_{k}\rangle_{d\lambda}}{\langle\pi_{k},\pi_{k}\rangle_{d\lambda}}=\frac{\langle t\pi_{k},\pi_{k}\rangle_{d\lambda}}{\Vert\pi_{k}\Vert_{d\lambda}^{2}},\quad k=0,1,2,\ldots,\\
\beta_{0} & =0,\\
\beta_{k} & =\frac{\langle\pi_{k},\pi_{k}\rangle_{d\lambda}}{\langle\pi_{k-1},\pi_{k-1}\rangle_{d\lambda}}=\frac{\Vert\pi_{k}\Vert_{d\lambda}^{2}}{\Vert\pi_{k-1}\Vert_{d\lambda}^{2}},\quad k=1,2,\ldots.
\end{align*}
\end{thm}

We now restrict our attention to Legendre polynomials, which are orthogonal
with respect to the measure $d\lambda(t)=dt$ on the interval $[-1,1]$.
The corresponding inner product is given by: 
\[
\langle f,g\rangle=\int_{-1}^{1}f(t)g(t)\,dt.
\]

\begin{thm}
\label{thm:legendre_alpha} For the monic Legendre polynomials $\{\pi_{k}(t)\}_{k=0}^{\infty}$,
orthogonal with respect to the inner product $\langle f,g\rangle=\int_{-1}^{1}f(t)g(t)\,dt$,
the recurrence coefficients $\alpha_{k}$ are all zero: 
\[
\alpha_{k}=\frac{\langle t\pi_{k},\pi_{k}\rangle}{\Vert\pi_{k}\Vert^{2}}=0,\quad k=0,1,2,\ldots.
\]
Furthermore, the monic Legendre polynomials exhibit alternating parity:
$\pi_{k}(-t)=(-1)^{k}\pi_{k}(t)$. Specifically, $\pi_{k}(t)$ is
an even function if $k$ is even, and an odd function if $k$ is odd.
\end{thm}

\begin{proof}
We prove that $\alpha_{k}=0$ and the parity property simultaneously
by induction. 
\begin{itemize}
\item \textbf{$k=0$:} $\pi_{0}(t)=1$, which is an even function. 
\[
\langle t\pi_{0},\pi_{0}\rangle=\int_{-1}^{1}t\,dt=\left[\frac{t^{2}}{2}\right]_{-1}^{1}=\frac{1}{2}-\frac{1}{2}=0.
\]
Since the numerator is zero and $\Vert\pi_{0}\Vert^{2}=\int_{-1}^{1}1^{2}\,dt=2>0$,
we have $\alpha_{0}=0$.
\item \textbf{$k=1$:} Using Theorem~\ref{thm:monic_recurrence} with $k=0$:
\[
\pi_{1}(t)=(t-\alpha_{0})\pi_{0}(t)-\beta_{0}\pi_{-1}(t)=(t-0)(1)-0=t.
\]
Thus, $\pi_{1}(t)=t$ is an odd function. Now compute $\alpha_{1}$:
\[
\langle t\pi_{1},\pi_{1}\rangle=\int_{-1}^{1}t^{3}\,dt=\left[\frac{t^{4}}{4}\right]_{-1}^{1}=\frac{1}{4}-\frac{1}{4}=0.
\]
Since $\Vert\pi_{1}\Vert^{2}=\int_{-1}^{1}t^{2}\,dt=\frac{2}{3}>0$,
we have $\alpha_{1}=0$.
\item \textbf{Inductive step:} Assume $\alpha_{i}=0$ and $\pi_{i}(-t)=(-1)^{i}\pi_{i}(t)$
for all $0\le i\le k$. We need to show $\alpha_{k+1}=0$ and that
$\pi_{k+1}(-t)=(-1)^{k+1}\pi_{k+1}(t)$.

From the recurrence (using $\alpha_{k}=0$), 
\[
\pi_{k+1}(t)=t\pi_{k}(t)-\beta_{k}\pi_{k-1}(t),\quad k\ge1.
\]
Applying the inductive hypothesis: 
\begin{align*}
\pi_{k+1}(-t) & =(-t)\pi_{k}(-t)-\beta_{k}\pi_{k-1}(-t)\\
 & =(-t)(-1)^{k}\pi_{k}(t)-\beta_{k}(-1)^{k-1}\pi_{k-1}(t)\\
 & =(-1)^{k+1}\left(t\pi_{k}(t)-\beta_{k}\pi_{k-1}(t)\right)=(-1)^{k+1}\pi_{k+1}(t).
\end{align*}

To compute $\alpha_{k+1}$, consider: 
\[
\langle t\pi_{k+1},\pi_{k+1}\rangle=\int_{-1}^{1}t[\pi_{k+1}(t)]^{2}\,dt.
\]
The integrand $t[\pi_{k+1}(t)]^{2}$ is odd: $t$ is odd, $[\pi_{k+1}(t)]^{2}$
is even. Thus, the integral over the symmetric interval $[-1,1]$
vanishes: 
\[
\int_{-1}^{1}t[\pi_{k+1}(t)]^{2}\,dt=0.
\]
Since $\Vert\pi_{k+1}\Vert^{2}>0$ (as $\pi_{k+1}(t)\neq0$ for all
$t\in[-1,1]$), we conclude that $\alpha_{k+1}=0$.

\end{itemize}
\end{proof}
\begin{lem}
\label{lem:a_k_lemma}Given $a_{1}=1$ and the recurrence relation
for $k\ge2$: 
\[
a_{k}=1-\frac{2k-3}{2k-1}a_{k-1},
\]
the explicit formula for $a_{k}$ is 
\[
a_{k}=\frac{k}{2k-1},\quad\text{for }k\ge1.
\]
\end{lem}

\begin{proof}
We prove this by induction. 
\begin{itemize}
\item \textbf{Base case ($k=1$):} 
\[
a_{1}=\frac{1}{2(1)-1}=\frac{1}{1}=1.
\]
\item \textbf{Inductive step:} Assume $a_{k-1}=\frac{k-1}{2(k-1)-1}=\frac{k-1}{2k-3}$
holds for some $k\ge2$. Substituting into the recurrence relation:
\begin{align*}
a_{k} & =1-\frac{2k-3}{2k-1}a_{k-1}=1-\frac{2k-3}{2k-1}\left(\frac{k-1}{2k-3}\right)=1-\frac{k-1}{2k-1}\\
 & =\frac{2k-1-(k-1)}{2k-1}=\frac{2k-1-k+1}{2k-1}=\frac{k}{2k-1}.
\end{align*}
\end{itemize}
\end{proof}
\begin{lem}
\label{lem:pi_k_P_k} Let $\{\pi_{k}(t)\}_{k=0}^{\infty}$ be a sequence
of monic orthogonal polynomials with respect to an inner product $\langle\cdot,\cdot\rangle_{d\lambda}$
(with corresponding norm $\Vert\cdot\Vert_{d\lambda}$). For any $k\ge0$,
if $P_{k}^{\text{M}}(t)$ is any monic polynomial of degree $k$,
then 
\[
\langle\pi_{k},P_{k}^{\text{M}}\rangle_{d\lambda}=\Vert\pi_{k}\Vert_{d\lambda}^{2}.
\]
\end{lem}

\begin{proof}
Let $P_{k}^{\text{M}}(t)$ be an arbitrary monic polynomial of degree
$k$. We can express $P_{k}^{\text{M}}(t)$ in the orthogonal basis
$\{\pi_{k}(t)\}_{k=0}^{\infty}$ as 
\[
P_{k}^{\text{M}}(t)=\pi_{k}(t)+\sum_{j=0}^{k-1}c_{j}\pi_{j}(t).
\]
By linearity of the inner product, we have 
\[
\langle\pi_{k},P_{k}^{\text{M}}\rangle_{d\lambda}=\langle\pi_{k},\pi_{k}\rangle_{d\lambda}+\sum_{j=0}^{k-1}c_{j}\langle\pi_{k},\pi_{j}\rangle_{d\lambda}.
\]
Since the polynomials are orthogonal, $\langle\pi_{k},\pi_{j}\rangle_{d\lambda}=0$
for $k\ne j$. Thus, 
\[
\langle\pi_{k},P_{k}^{\text{M}}\rangle_{d\lambda}=\Vert\pi_{k}\Vert_{d\lambda}^{2}.
\]
\end{proof}
\begin{thm}[\cite{AbramowitzStegun1965,Szego1975}]
\label{thm:legendre_norm_value_revised} For the monic Legendre polynomials
$\pi_{k}(t)$, the following hold: 
\begin{enumerate}
\item \textbf{Relationship between norm and value at $t=1$:} 
\begin{equation}
\Vert\pi_{k}\Vert^{2}=\frac{2\pi_{k}^{2}(1)}{2k+1},\quad k\ge0.\label{eq:norm_vs_pi_k_1_revised}
\end{equation}
\item \textbf{Explicit value at $t=1$:} 
\begin{equation}
\pi_{k}(1)=\frac{2^{k}(k!)^{2}}{(2k)!},\quad k\ge0.\label{eq:pi_k_at_1_revised}
\end{equation}
\item \textbf{Explicit squared norm:} 
\begin{equation}
\Vert\pi_{k}\Vert^{2}=\frac{2^{2k+1}(k!)^{4}}{(2k+1)\,((2k)!)^{2}},\quad k\ge0.\label{eq:monic_legendre_norm_revised}
\end{equation}
\end{enumerate}
\end{thm}

\begin{proof}
We first derive the relationship \eqref{eq:norm_vs_pi_k_1_revised},
then find the explicit formula for $\pi_{k}(1)$ \eqref{eq:pi_k_at_1_revised},
and finally combine them to obtain the explicit norm \eqref{eq:monic_legendre_norm_revised}.

\textbf{Part 1: Relationship between $\Vert\pi_{k}\Vert^{2}$ and
$\pi_{k}(1)$.} We apply integration by parts to $\Vert\pi_{k}\Vert^{2}$:
\begin{equation}
\Vert\pi_{k}\Vert^{2}=\int_{-1}^{1}\pi_{k}(t)^{2}\,dt=\left[t\pi_{k}^{2}(t)\right]_{-1}^{1}-\int_{-1}^{1}t\cdot2\pi_{k}(t)\pi'_{k}(t)\,dt.\label{eq:per_partes}
\end{equation}
From Theorem~\ref{thm:legendre_alpha}, we know that $\pi_{k}(-t)=(-1)^{k}\pi_{k}(t)$,
implying $\pi_{k}^{2}(-1)=\pi_{k}^{2}(1)$. The boundary term evaluates
to: 
\[
\left[t\pi_{k}^{2}(t)\right]_{-1}^{1}=\pi_{k}^{2}(1)+\pi_{k}^{2}(1)=2\pi_{k}^{2}(1).
\]
If $k=0$, then 
\[
\int_{-1}^{1}t\cdot2\pi_{k}(t)\pi'_{k}(t)\,dt=0,
\]
so $\Vert\pi_{k}\Vert^{2}=2\pi_{k}^{2}(1)$. If $k>0$, then $t\pi'_{k}(t)=kP_{k}^{\text{M}}(t)$
for a monic polynomial $P_{k}^{\text{M}}(t)$ of degree $k$, since
the leading term of $\pi'_{k}$ is $kt^{k-1}$. By Lemma~\ref{lem:pi_k_P_k},
\[
\int_{-1}^{1}2t\pi_{k}(t)\pi'_{k}(t)\,dt=2k\int_{-1}^{1}\pi_{k}(t)P_{k}^{\text{M}}(t)\,dt=2k\Vert\pi_{k}\Vert^{2}.
\]
Substituting into \eqref{eq:per_partes}, we obtain: 
\begin{equation}
\Vert\pi_{k}\Vert^{2}=2\pi_{k}^{2}(1)-2k\Vert\pi_{k}\Vert^{2},\quad\forall k\ge0.\label{eq:norm_ibp_rel_revised}
\end{equation}
Solving gives: 
\[
\Vert\pi_{k}\Vert^{2}=\frac{2\pi_{k}^{2}(1)}{2k+1},\quad\forall k\ge0.
\]

\textbf{Part 2: Explicit formula for $\pi_{k}(1)$.} We know that
$\pi_{0}(t)=1$ and $\pi_{1}(t)=t$, so $\pi_{0}(1)=\pi_{1}(1)=1$.
From the recurrence relation at $t=1$: 
\[
\pi_{k}(1)=\pi_{k-1}(1)-\beta_{k-1}\pi_{k-2}(1),\quad\forall k\ge2,
\]
where $\beta_{k-1}=\frac{\Vert\pi_{k-1}\Vert^{2}}{\Vert\pi_{k-2}\Vert^{2}}$
(Theorem~\ref{thm:monic_recurrence}). Using the result from Part
1: 
\[
\Vert\pi_{j}\Vert^{2}=\frac{2\pi_{j}^{2}(1)}{2j+1}.
\]
We substitute this into the recurrence: 
\begin{align*}
\pi_{k}(1) & =\pi_{k-1}(1)-\frac{\frac{2\pi_{k-1}^{2}(1)}{2k-1}}{\frac{2\pi_{k-2}^{2}(1)}{2k-3}}\pi_{k-2}(1)=\pi_{k-1}(1)-\frac{2k-3}{2k-1}\cdot\frac{2\pi_{k-1}^{2}(1)}{2\pi_{k-2}(1)}.
\end{align*}
Letting $a_{k}=\pi_{k}(1)/\pi_{k-1}(1)$, we derive the recurrence:
\[
a_{k}=1-\frac{2k-3}{2k-1}a_{k-1},\quad k\ge2,
\]
with $a_{1}=1$. By Lemma~\ref{lem:a_k_lemma}, this gives: 
\[
a_{k}=\frac{k}{2k-1}.
\]
Hence, 
\[
\pi_{k}(1)=\prod_{i=1}^{k}a_{i}=\prod_{i=1}^{k}\frac{i}{2i-1}=\frac{k!}{1\cdot3\cdot5\cdots(2k-1)}.
\]
Using the identity 
\[
1\cdot3\cdot5\cdots(2k-1)=\frac{1\cdot2\cdot3\cdot4\cdot5\cdot6\cdots(2k)}{2\cdot4\cdot6\cdots(2k)}=\frac{(2k)!}{2^{k}k!},
\]
 we obtain: 
\[
\pi_{k}(1)=\frac{k!}{\frac{(2k)!}{2^{k}k!}}=\frac{2^{k}(k!)^{2}}{(2k)!},\quad k\ge0.
\]

\textbf{Part 3: Explicit squared norm.} Substituting \eqref{eq:pi_k_at_1_revised}
into \eqref{eq:norm_vs_pi_k_1_revised} gives: 
\[
\Vert\pi_{k}\Vert^{2}=\frac{2}{2k+1}\left(\frac{2^{k}(k!)^{2}}{(2k)!}\right)^{2}.
\]
Simplifying: 
\[
\Vert\pi_{k}\Vert^{2}=\frac{2}{2k+1}\cdot\frac{2^{2k}(k!)^{4}}{((2k)!)^{2}}=\frac{2^{2k+1}(k!)^{4}}{(2k+1)((2k)!)^{2}},\quad k\ge0.
\]
\end{proof}
\begin{thm}
\label{thm:legendre_beta} For the monic Legendre polynomials, the
recurrence coefficient $\beta_{k}$ for $k\ge1$ is given by 
\[
\beta_{k}=\frac{k^{2}}{4k^{2}-1}.
\]
\end{thm}

\begin{proof}
Starting from the definition in Theorem~\ref{thm:monic_recurrence}:
\[
\beta_{k}=\frac{\Vert\pi_{k}\Vert^{2}}{\Vert\pi_{k-1}\Vert^{2}},\quad k\ge1.
\]
Using the result from Theorem~\ref{thm:legendre_norm_value_revised}:
\[
\Vert\pi_{k}\Vert^{2}=\frac{2^{2k+1}(k!)^{4}}{(2k+1)((2k)!)^{2}},
\]
and for $k-1$: 
\[
\Vert\pi_{k-1}\Vert^{2}=\frac{2^{2k-1}((k-1)!)^{4}}{(2k-1)((2k-2)!)^{2}}.
\]
Now compute the ratio: 
\begin{align*}
\beta_{k} & =\frac{\Vert\pi_{k}\Vert^{2}}{\Vert\pi_{k-1}\Vert^{2}}=\frac{\frac{2^{2k+1}(k!)^{4}}{(2k+1)((2k)!)^{2}}}{\frac{2^{2k-1}((k-1)!)^{4}}{(2k-1)((2k-2)!)^{2}}}=\frac{2^{2k+1}(k!)^{4}}{(2k+1)((2k)!)^{2}}\cdot\frac{(2k-1)((2k-2)!)^{2}}{2^{2k-1}((k-1)!)^{4}}\\
 & =\frac{2^{2k+1}}{2^{2k-1}}\cdot\frac{(k!)^{4}}{((k-1)!)^{4}}\cdot\frac{2k-1}{2k+1}\cdot\frac{((2k-2)!)^{2}}{((2k)!)^{2}}\\
 & =4\cdot k^{4}\cdot\frac{2k-1}{2k+1}\cdot\frac{1}{(2k)^{2}(2k-1)^{2}}=\frac{k^{2}}{4k^{2}-1}.
\end{align*}
\end{proof}
\begin{cor}[Recurrence for Monic Legendre Polynomials]
\label{cor:The-three-term-recurrence-legendre}The three-term recurrence
relation for the monic Legendre polynomials $\pi_{k}(t)$ is 
\[
\pi_{k+1}(t)=t\pi_{k}(t)-\frac{k^{2}}{4k^{2}-1}\pi_{k-1}(t),\quad k=1,2,\ldots
\]
with initial conditions $\pi_{0}(t)=1$ and $\pi_{1}(t)=t$. 
\end{cor}

\subsubsection{Three-term recurrence relation for normalized Legendre polynomials}\label{subsec:Three-term-recurrence-relation-normalized}

Orthonormal polynomials $\tilde{\pi}_{k}(t)$ are obtained by normalizing
the monic orthogonal polynomials: 
\[
\tilde{\pi}_{k}(t)=\frac{\pi_{k}(t)}{\Vert\pi_{k}\Vert_{d\lambda}},\quad\text{so that }\langle\tilde{\pi}_{k},\tilde{\pi}_{j}\rangle_{d\lambda}=\delta_{kj}.
\]
They also satisfy a three-term recurrence relation, which can be derived
from the recurrence for the monic polynomials. 
\begin{thm}[Orthonormal Legendre polynomials]
\label{Orthonormal-Legendre-polynomials}The orthonormal Legendre
polynomials $\tilde{\pi}_{k}(t)$ satisfy the three-term recurrence
relation: 
\begin{equation}
\tilde{\pi}_{k+1}(t)=\sqrt{4-\frac{1}{(k+1)^{2}}}\left(t\tilde{\pi}_{k}(t)-\frac{1}{\sqrt{4-\frac{1}{k^{2}}}}\tilde{\pi}_{k-1}(t)\right),\quad k=1,2,\ldots\label{eq:orthonormal_legendre}
\end{equation}
with initial conditions $\tilde{\pi}_{0}(t)=1/\Vert\pi_{0}\Vert=1/\sqrt{2}$
and $\tilde{\pi}_{1}(t)=t/\Vert\pi_{1}\Vert=\sqrt{\frac{3}{2}}\,t$.
\end{thm}

\begin{proof}
Let 
\[
a_{k}:=\Vert\pi_{k}\Vert\quad\Longrightarrow\quad\tilde{\pi}_{k}(t)=\frac{\pi_{k}(t)}{a_{k}},\qquad k\ge0.
\]
The monic Legendre relation (Corollary~\ref{cor:The-three-term-recurrence-legendre})
divided by $a_{k+1}$, reads
\begin{equation}
\tilde{\pi}_{k+1}(t)=\frac{a_{k}}{a_{k+1}}\left(t\,\tilde{\pi}_{k}(t)-\frac{k^{2}}{4k^{2}-1}\cdot\frac{a_{k-1}}{a_{k}}\tilde{\pi}_{k-1}(t)\right),\qquad k\ge1.\label{eq:tilde_start}
\end{equation}
From the explicit norm formula ${\displaystyle a_{k}^{\,2}=\frac{2^{2k+1}(k!)^{4}}{(2k+1)\,((2k)!)^{2}}}$
(Theorem~\ref{thm:legendre_norm_value_revised}) one finds 
\begin{align*}
\frac{a_{k}^{\,2}}{a_{k+1}^{\,2}} & =\frac{2^{2k+1}(k!)^{4}}{(2k+1)\,((2k)!)^{2}}\cdot\frac{(2k+3)\,((2k+2)!)^{2}}{2^{2k+3}((k+1)!)^{4}}=\frac{4\,(k+1)^{2}\,(2k+1)\,(2k+3)}{4(k+1)^{4}}\\
 & =\frac{4k^{2}+8k+3}{(k+1)^{2}}=4-\frac{1}{(k+1)^{2}},\\
\frac{a_{k-1}^{\,2}}{a_{k+1}^{\,2}} & =\frac{a_{k-1}^{2}}{a_{k}^{2}}\cdot\frac{a_{k}^{2}}{a_{k+1}^{2}}.
\end{align*}
Hence 
\begin{align*}
\frac{a_{k}}{a_{k+1}} & =\sqrt{4-\frac{1}{(k+1)^{2}}},\\
\frac{k^{2}}{4k^{2}-1}\frac{a_{k-1}}{a_{k}} & =\frac{k^{2}}{4k^{2}-1}\sqrt{4-\frac{1}{k^{2}}}=\frac{k^{2}}{4k^{2}-1}\sqrt{\frac{4k^{2}-1}{k^{2}}}=\frac{k\sqrt{4k^{2}-1}}{4k^{2}-1}=\frac{1}{\sqrt{4-\frac{1}{k^{2}}}}.
\end{align*}
Substituting the two ratios gives 
\[
\tilde{\pi}_{k+1}(t)=\sqrt{4-\frac{1}{(k+1)^{2}}}\,\biggl(t\,\tilde{\pi}_{k}(t)-\frac{1}{\sqrt{4-\frac{1}{k^{2}}}}\,\tilde{\pi}_{k-1}(t)\biggr),\qquad k\ge1,
\]
which is exactly \eqref{eq:orthonormal_legendre}. The initial values
follow from $\tilde{\pi}_{0}(t)=\pi_{0}(t)/a_{0}=1/\sqrt{2}$ and
$\tilde{\pi}_{1}(t)=t/a_{1}=t\sqrt{3/2}$. 
\end{proof}

\subsubsection{Three-term recurrence relation for classical Legendre polynomials}\label{subsec:Three-term-recurrence-relation-classical}

We define the \emph{classical} Legendre polynomials $P_{k}$ by 
\begin{equation}
d_{k}\,P_{k}(t)=\pi_{k}(t)\quad\Longrightarrow\quad P_{k}(1)=1,\qquad k\ge0,\label{eq:def_classical}
\end{equation}
where 
\begin{equation}
d_{k}:=\pi_{k}(1)=\frac{2^{k}(k!)^{2}}{(2k)!},\qquad k\ge0,\label{eq:d_k_pi_k}
\end{equation}
see Theorem~\ref{thm:legendre_norm_value_revised}.
\begin{thm}
For every $k\ge1$ the polynomials defined in~\eqref{eq:def_classical}
satisfy 
\[
(k+1)\,P_{k+1}(t)\;=\;(2k+1)\,t\,P_{k}(t)\;-\;k\,P_{k-1}(t),\qquad P_{0}(t)=1,\;P_{1}(t)=t.
\]
\end{thm}

\begin{proof}
We start the proof by substituting~\eqref{eq:def_classical} into
the \emph{monic} Legendre three-term recurrence (Corollary~\ref{cor:The-three-term-recurrence-legendre}):
\[
d_{k+1}P_{k+1}(t)\;=\;t\,d_{k}P_{k}(t)\;-\;\frac{k^{2}}{4k^{2}-1}\,d_{k-1}P_{k-1}(t).
\]
Dividing by $d_{k+1}$ gives 
\[
P_{k+1}(t)\;=\;\frac{d_{k}}{d_{k+1}}\,t\,P_{k}(t)\;-\;\frac{k^{2}}{4k^{2}-1}\frac{d_{k-1}}{d_{k+1}}\,P_{k-1}(t).
\]
Using \eqref{eq:d_k_pi_k} one finds 
\[
\frac{d_{k}}{d_{k+1}}=\frac{(2k+1)(2k+2)}{2(k+1)^{2}}=\frac{2k+1}{k+1},\qquad\frac{d_{k-1}}{d_{k+1}}=\frac{(2k+1)(2k-1)}{k(k+1)}.
\]
(The first equality comes from 
\[
\frac{(2(k+1))!}{2^{k+1}((k+1)!)^{2}}=\frac{(2k)!}{2^{k}(k!)^{2}}\;\frac{(2k+1)(2k+2)}{2(k+1)^{2}},
\]
and multiplying it by $\tfrac{2k(2k-1)}{2k^{2}}$ yields the second.)

\noindent Because 
\[
\frac{k^{2}}{4k^{2}-1}\,\frac{d_{k-1}}{d_{k+1}}=\frac{k^{2}}{(2k+1)(2k-1)}\frac{(2k+1)(2k-1)}{k(k+1)}=\frac{k}{k+1},
\]
we obtain 
\[
P_{k+1}(t)=\frac{2k+1}{k+1}\,t\,P_{k}(t)-\frac{k}{k+1}\,P_{k-1}(t).
\]
Multiplying through by $k+1$ gives the claimed recurrence. 
\end{proof}

\subsection{Construction of an $n$-Point Gaussian Quadrature Rule}

\label{subsec:GW_quadrature}

Let $\{\pi_{k}\}_{k=0}^{\infty}$ be the monic orthogonal polynomials
from Theorem~\ref{thm:monic_recurrence}, with recurrence coefficients
$\{\alpha_{k}\}_{k\ge0}$ and $\{\beta_{k}\}_{k\ge1}$. Following~\cite{GolubWelsch1969},
define the \emph{Jacobi matrix} 
\[
J_{n}=\begin{pmatrix}\alpha_{0} & \sqrt{\beta_{1}}\\
\sqrt{\beta_{1}} & \alpha_{1} & \ddots\\
 & \ddots & \ddots & \sqrt{\beta_{n-1}}\\
 &  & \sqrt{\beta_{n-1}} & \alpha_{n-1}
\end{pmatrix}\in\mathbb{R}^{n\times n},
\]
which is real, symmetric, and tridiagonal. The following holds: 
\begin{itemize}
\item The eigenvalues $\{t_{i}\}_{i=1}^{n}$ of $J_{n}$ are \emph{exactly}
the nodes of the unique $n$-point Gaussian quadrature rule for the
measure $d\lambda$ (and also the roots of $\pi_{n}$). That is, $\{t_{i}\}_{i=1}^{n}\subset[a,b]$,
and the node sets for the $n$-point and $(n+1)$-point Gaussian quadrature
satisfy the alternating property \cite[Thm. 1.20]{Gautschi2004}. 
\item Let $v^{(i)}=(v_{1}^{(i)},\dots,v_{n}^{(i)})^{\top}$ be a \emph{unit}
eigenvector of $J_{n}$ associated with $t_{i}$. Then the corresponding
positive weight is 
\[
w_{i}=m_{0}\bigl(v_{1}^{(i)}\bigr)^{2}>0,\qquad i=1,\dots,n,\quad m_{0}=\int_{a}^{b}d\lambda(t).
\]
\end{itemize}
Because $J_{n}$ is symmetric and tridiagonal, its eigenvalues and
eigenvectors can be computed in $\mathcal{O}(n^{2})$ flops with high
relative accuracy via the implicit-shift symmetric QR algorithm, as
analyzed in~\cite{GolubVanLoan2013}. SciPy’s implementation \verb|eigh_tridiagonal|
uses the LAPACK routines \texttt{DSTEVD} and \texttt{DSTEV}, which
realize this algorithm.

\section{Karhunen--Loève Expansion of Stationary Isotropic Gaussian Random
Fields}\label{sec:Karhunen-Loeve-expansion-of}

Gaussian random fields (GRFs) are fundamental tools for modeling spatially
varying quantities subject to uncertainty across numerous scientific
and engineering disciplines \cite{AdlerTaylor2007}. A key challenge
lies in the computational representation of these infinite-dimensional
objects. The Karhunen--Loève (KL) expansion provides an optimal representation
in the sense that it minimizes the mean-squared truncation error for
a given number of terms \cite{Loeve1978}. This chapter focuses on
the KL expansion for a specific yet widely applicable class of GRFs:
zero-mean, stationary, and isotropic fields defined on a $D$-dimensional
hyper-rectangular domain $\Omega\subset\mathbb{R}^{D}$. 

\subsection{Gaussian Random Fields and Autocovariance}

Let $(\mathcal{A},\mathcal{F},\mathbb{P})$ be a complete probability
space, and let $\Omega=[a_{1},b_{1}]\times\dots\times[a_{D},b_{D}]\subset\mathbb{R}^{D}$
be a $D$-dimensional closed interval (hyper-rectangle). We consider
a real-valued Gaussian random field (GRF) $Z(\bm{x},\omega):\Omega\times\mathcal{A}\to\mathbb{R}$.
In the following, $Z(\bm{x})\sim\mathcal{N}(\mu,\sigma)$ denotes
a random variable representing the behavior of the field at a specific
spatial point $\bm{x}\in\Omega$. We assume $Z$ has zero mean: 
\begin{equation}
\mathbb{E}[Z(\bm{x})]=0,\quad\forall\bm{x}\in\Omega.
\end{equation}

We further assume that $Z$ is a \emph{second-order} random field,
meaning it has finite second moments at every spatial location: 
\begin{equation}
\mathbb{E}\bigl[Z(\bm{x})^{2}\bigr]<\infty,\quad\forall\bm{x}\in\Omega.
\end{equation}
This condition ensures that the covariance function introduced below
is well defined.

The field is characterized by its autocovariance function $C:\Omega\times\Omega\to\mathbb{R}$,
defined as: 
\begin{equation}
C(\bm{x},\bm{y})=\mathbb{E}[Z(\bm{x})Z(\bm{y})],\quad\forall\bm{x},\bm{y}\in\Omega.
\end{equation}

We assume $Z$ is stationary, meaning $C(\bm{x},\bm{y})$ depends
only on the separation vector $\bm{x}-\bm{y}$: 
\begin{equation}
C(\bm{x},\bm{y})=C_{S}(\bm{x}-\bm{y}).
\end{equation}

Furthermore, we assume $Z$ is isotropic, meaning $C_{S}(\bm{x}-\bm{y})$
depends only on the Euclidean norm of the separation vector, $\|\bm{x}-\bm{y}\|_{2}=\sqrt{\sum_{l=1}^{D}(x_{l}-y_{l})^{2}}$:
\begin{equation}
C(\bm{x},\bm{y})=C_{I}(\|\bm{x}-\bm{y}\|_{2}),
\end{equation}
where $C_{I}:[0,\infty)\to\mathbb{R}$ is a positive-semidefinite
function \cite{ChilesDelfiner2012}. Later in this paper, we will
drop the subscripts and distinguish between $C$, $C_{S}$, and $C_{I}$
based on the type of argument.

We assume that the sample paths of $Z$ are square-integrable over
$\Omega$, i.e., $Z(\cdot,\omega)\in L^{2}(\Omega)$ almost surely.
This is guaranteed if the covariance function $C$ satisfies $\int_{\Omega}C(\bm{x},\bm{x})\,d\bm{x}<\infty$.
For stationary fields on bounded domains, $C(\bm{x},\bm{x})=C_{I}(0)=\sigma^{2}$
(the variance), which is finite and ensures this condition holds.

The autocovariance function defines a linear integral operator $\mathcal{C}:L^{2}(\Omega)\to L^{2}(\Omega)$:
\begin{equation}
(\mathcal{C}f)(\bm{x})=\int_{\Omega}C(\bm{x},\bm{y})f(\bm{y})\,d\bm{y}.\label{eq:cov_operator}
\end{equation}

\subsection{The Karhunen--Loève Expansion Theorem}

The KL expansion provides a representation of the random field $Z(\bm{x})$
in terms of a deterministic orthonormal basis $\{\phi_{j}(\bm{x})\}_{j=1}^{\infty}$
of $L^{2}(\Omega)$ and a sequence of uncorrelated random variables
$\{\xi_{j}(\omega)\}_{j=1}^{\infty}$.
\begin{thm}[Karhunen--Loève \cite{Loeve1978}, \cite{GhanemSpanos1991}]
\label{thm:kl} Let $Z(\bm{x})$ be a zero-mean, second-order random
field on $\Omega$ with continuous autocovariance function $C(\bm{x},\bm{y})$,
such that $Z(\cdot,\omega)\in L^{2}(\Omega)$ almost surely. Let $(\lambda_{j},\phi_{j})_{j=1}^{\infty}$
be the eigenpairs of the autocovariance operator $\mathcal{C}$: 
\begin{equation}
\mathcal{C}\phi_{j}=\lambda_{j}\phi_{j},\quad\text{i.e.,}\quad\int_{\Omega}C(\bm{x},\bm{y})\phi_{j}(\bm{y})\,d\bm{y}=\lambda_{j}\phi_{j}(\bm{x}).\label{eq:kl_eigenproblem}
\end{equation}
Here, $\{\phi_{j}\}_{j=1}^{\infty}$ forms an orthonormal basis for
the closure of the range of $\mathcal{C}$ (and can be extended to
an orthonormal basis of $L^{2}(\Omega)$). Then $Z(\bm{x})$ admits
the following expansion: 
\begin{equation}
Z(\bm{x},\omega)=\sum_{j=1}^{\infty}\sqrt{\lambda_{j}}\,\xi_{j}(\omega)\phi_{j}(\bm{x}),\label{eq:kl_expansion}
\end{equation}
where the convergence is in the mean-square sense, uniformly in $\bm{x}\in\Omega$.

The random variables $\xi_{j}(\omega)$ are given by 
\begin{equation}
\xi_{j}(\omega)=\frac{1}{\sqrt{\lambda_{j}}}\int_{\Omega}Z(\bm{x},\omega)\phi_{j}(\bm{x})\,d\bm{x},\quad(\text{for }\lambda_{j}>0)
\end{equation}
and satisfy $\mathbb{E}[\xi_{j}]=0$ and $\mathbb{E}[\xi_{j}\xi_{m}]=\delta_{jm}$
(Kronecker delta). If $Z$ is a Gaussian random field, then the $\xi_{j}$
are independent standard Gaussian random variables. 
\end{thm}

\section{Approximation of Autocovariance Functions via Gaussian Sums}\label{sec:gauss_sum}

Let $d=\|\bm{x}-\bm{y}\|_{2}$ denote the Euclidean distance between
two points $\bm{x},\bm{y}\in\mathbb{R}^{D}$. We consider stationary,
isotropic autocovariance functions $C(d)$ defined for $d\ge0$. For
random fields in two or more dimensions $D$, it is highly desirable
to express the autocovariance in a form that allows separation of
individual dimensions, i.e., 
\[
C(d)=\sum_{i=1}^{k}a_{i}\prod_{l=1}^{D}C_{il}(d_{l}^{2}),\quad d_{l}^{2}=\left(x_{l}-y_{l}\right)^{2}.
\]
Such a form of autocovariance will be very useful later for Galerkin
discretization in subsection \ref{sec:galerkin}. We investigate the
approximation of $C(d)$ using finite sums of Gaussian functions:
\begin{equation}
C_{k}(d)=\sum_{i=1}^{k}a_{i}e^{-b_{i}d^{2}}\label{eq:gaussian_sum}
\end{equation}
where $a_{i}\in\mathbb{R}$ are coefficients of the basis functions
and $b_{i}>0$ are their scale parameters. Note that $d^{2}=\sum_{l=1}^{D}d_{l}^{2}$
and therefore 
\[
e^{-b_{i}d^{2}}=\prod_{l=1}^{D}e^{-b_{i}d_{l}^{2}},\quad\forall b_{i}.
\]

\subsection{General Approximability}

We first establish the general conditions under which a function $C(d)$
can be approximated arbitrarily well by sums of the form \eqref{eq:gaussian_sum}
on a compact interval $[0,L]$ for some $L>0$. This interval is typically
chosen large enough to capture the significant decay of the covariance
function. The natural space for measuring approximation error is often
the Hilbert space $L^{2}(0,L)$ equipped with the norm $\|f\|_{L^{2}}=(\int_{0}^{L}|f(d)|^{2}dd)^{1/2}$.
\begin{thm}[{Stone--Weierstrass \cite[Thm. 7.32]{Rudin1976}}]
\label{Stone--Weierstrass}Let $K$ be a compact interval $[0,L]$
for some $L>0$ and let $\mathcal{A}$ be an \emph{algebra} of real--valued
continuous functions on $K$; that is, $\mathcal{A}\subset C([0,L])$
is closed under linear combinations and pointwise products. If
\begin{itemize}
\item $\mathcal{A}$ \emph{separates points}: for every $x\neq y$ in $K$
there exists $f\in\mathcal{A}$ with $f(x)\neq f(y)$, and
\item $\mathcal{A}$ \emph{vanishes at no point}: for every $x\in K$ there
exists $f\in\mathcal{A}$ with $f(x)\neq0$,
\end{itemize}
then the uniform (supremum--norm) closure $\overline{\mathcal{A}}$
equals the whole space $C([0,L])$.
\end{thm}

\begin{prop}
\label{prop:density} Let $V=\operatorname{span}\{g_{b}(d):=e^{-bd^{2}}:b>0\}$.
Then, the vector space $V$ is dense in the space of continuous functions
$C([0,L])$ equipped with the supremum norm $\|\cdot\|_{\infty}$.
Consequently, $V$ is also dense in $L^{2}(0,L)$.
\end{prop}

\begin{proof}
We verify the assumptions of the Stone--Weierstrass theorem (theorem
\ref{Stone--Weierstrass}) for 
\[
\mathcal{A}:=\operatorname{span}\Bigl(\{1\}\cup\{g_{b}:b>0\}\Bigr).
\]
Clearly $V\subset\mathcal{A}$ and as $b\rightarrow0$, $g_{b}\left(d\right)\rightarrow1$,
$\overline{V}=\overline{\mathcal{A}}$. So it suffices to show $\overline{\mathcal{A}}=C([0,L])$.
\begin{itemize}
\item $\mathcal{A}$ is an algebra: Given finite sums $f=\sum_{i=1}^{m}a_{i}g_{b_{i}}$
and $h=\sum_{j=1}^{n}c_{j}g_{c_{j}}$, their product is 
\[
fh=\sum_{i,j}a_{i}c_{j}e^{-(b_{i}+c_{j})d^{2}}\in\mathcal{A},
\]
so $\mathcal{A}$ is closed under pointwise multiplication (and, by
being defined as span, also closed under addition and scalar multiplication).
\item $\mathcal{A}$ separates points: Let $x\neq y$ in $[0,L]$. Because
$x^{2}\neq y^{2}$, we have $g_{b}(x)=e^{-bx^{2}}\neq e^{-by^{2}}=g_{b}(y)$
for every $b>0$. Hence a single $g_{b}\in\mathcal{A}$ distinguishes
$x$ and $y$.
\item $\mathcal{A}$ vanishes at no point: For any fixed $d\in[0,L]$ and
any $b>0$, $g_{b}(d)=e^{-bd^{2}}>0$; thus every point has a non-zero
value under every $g_{b}$, and in particular under some element of
$\mathcal{A}$.
\end{itemize}
Then, by Stone--Weierstrass, supremum--norm closure $\overline{\mathcal{A}}$
equals $C([0,L])$. Since convergence in the supremum norm on $[0,L]$
implies convergence in the $L^{2}(0,L)$ norm (specifically, $\|f\|_{L^{2}}\le\sqrt{L}\|f\|_{\infty}$),
$V$ is dense in $C([0,L])$ with the $L^{2}$ topology. As $C([0,L])$
is itself dense in $L^{2}(0,L)$, it follows by the triangle inequality
that $V$ is dense in $L^{2}(0,L)$.
\end{proof}
This proposition guarantees that any autocovariance function $C(d)$
that is continuous on $[0,L]$ (and thus belongs to $L^{2}(0,L)$)
can be approximated with arbitrary precision in the $L^{2}$ sense
by a finite sum $C_{k}(d)$ as defined in \eqref{eq:gaussian_sum}. 

\subsection{Conditions for Non-Negative Coefficients}\label{subsec:Conditions-for-Non-Negative}

While Proposition \ref{prop:density} guarantees the existence of
an approximation, it does not specify the signs of the coefficients
$a_{i}$. However, we require the approximation of the autocovariance
to be positive semidefinite; thus, the coefficients $a_{i}$ must
be positive. The sign structure of $a_{i}$ is intimately linked to
the concept of complete monotonicity.
\begin{defn}[Complete Monotonicity]
 A function $f(x)$ defined on $(0,\infty)$ is \emph{completely
monotonic} if it is of class $C^{\infty}$ and satisfies $(-1)^{k}f^{(k)}(x)\ge0$
for all $x>0$ and integers $k=0,1,2,\dots$. 
\end{defn}

Bernstein's theorem provides an alternative definition of complete
monotonicity of functions as Laplace transforms of non-negative measures,
which will be very useful to link to the square exponential approximation.
\begin{thm}[{Bernstein's \cite[Thm. 12a]{Widder1941}}]
\label{thm:bernstein} A necessary and sufficient condition that
$f(x)$ should be completely monotonic in $0\le x<\infty$ is that
\begin{equation}
f(x)=\int_{0}^{\infty}e^{-xt}d\alpha(t),
\end{equation}
where $\alpha(t)$ is bounded and non-decreasing and the integral
converges for $0\le x<\infty$.
\end{thm}

Now let $x=d^{2}$ and set $f(x)=C(\sqrt{x})$. Suppose $f$ is completely
monotone on $(0,\infty)$, or at least on $(0,L^{2}]$. By Bernstein’s
theorem (Theorem~\ref{thm:bernstein}), $f$ admits the Laplace representation
\begin{equation}
C(d)=f(d^{2})=\int_{0}^{\infty}e^{-d^{2}t}\,d\mu(t),\label{eq:cm_representation}
\end{equation}
where $\mu$ is a non-negative, bounded, and non-decreasing measure.
The integral in~\eqref{eq:cm_representation} can be approximated
by the finite Gaussian sum 
\begin{equation}
C(d)\;\approx\;\sum_{i=1}^{k}w_{i}e^{-d^{2}t_{i}}.\label{eq:gaussian_sum-1}
\end{equation}

Because $\mu\ge0$, it induces a bilinear form on the space of polynomials:
\[
\langle p,q\rangle_{\mu}:=\int_{0}^{\infty}p(t)q(t)\,d\mu(t).
\]
If $\mu$ has infinite support, this form is strictly positive and
thus defines an inner product on the space of polynomials. In that
case,~\eqref{eq:gaussian_sum-1} can be interpreted as a Gaussian
quadrature rule with nodes $b_{i}=t_{i}>0$ and corresponding weights
$a_{i}=w_{i}>0$. 

Classical error bounds for Gaussian quadrature with analytic integrands
(see, e.g., \cite[Ch.~2]{Gautschi2004}) imply that the error in~\eqref{eq:gaussian_sum-1}
decays \emph{exponentially} with~$k$. Hence the squared-exponential
approximation should converge exponentially fast.

If $\mu$ is supported on finitely many points, $d\mu(t)=\sum_{i=1}^{k}w_{i}\delta_{t_{i}}(dt)$
with $a_{i}=w_{i}>0$ and $b_{i}=t_{i}>0$, then 
\[
C(d)=\sum_{i=1}^{k}a_{i}e^{-d^{2}b_{i}},
\]
so $C$ is already a finite Gaussian mixture and no further approximation
is necessary.

\subsubsection{Non-negative coefficients for isotropic autocovariance functions}

We now examine which \emph{isotropic} autocovariance functions can
be approximated by a mixture of Gaussians (squared-exponentials) with
\emph{non-negative} coefficients~$a_{i}$ in~\eqref{eq:gaussian_sum}.
\begin{thm}[{Schoenberg {\cite[Thm.~3]{Schoenberg1938}}, {\cite[Thm.~7.14]{Wendland2004}}}]
\label{thm:schoenberg_equiv} For a function $C\colon[0,\infty)\to\mathbb{R}$
the following two properties are equivalent: 
\begin{enumerate}
\item[(i)] $C(\|x\|)$ is positive definite on \emph{every} $\mathbb{R}^{D}$,
$D\in\mathbb{N}$; 
\item[(ii)] the radial profile $C(\sqrt{\cdot})$ is completely monotone on $[0,\infty)$
and not constant. 
\end{enumerate}
\end{thm}

Thus any isotropic kernel $C$ that is positive definite in every
dimension, including the Matérn, powered-exponential, Cauchy, and
rational-quadratic families used throughout spatial statistics, has
a completely monotone radial profile. Consequently, $C$ admits the
Bernstein-Widder representation~\eqref{eq:cm_representation}. Restricting
that integral to a bounded interval $[0,L]$ and replacing it by an
$k$-point Gauss quadrature yields the finite Gaussian mixture \eqref{eq:gaussian_sum}
with weights $a_{i}=w_{i}>0$ and nodes $b_{i}=t_{i}>0$. The requirement
$a_{i}>0$ is therefore \emph{not} an additional constraint, it is
the structural property that guarantees positive definiteness in every
dimension and is automatically satisfied for all kernels of practical
interest.

The only isotropic kernels that fail to be strictly positive definite
are the degenerate cases (e.g.\ the constant kernel $C(d)\equiv\sigma^{2}$
or the identically zero kernel), which correspond to random fields
with zero variance in some non-trivial linear combination of observations.
Excluding these trivial exceptions, every isotropic covariance function
encountered in practice fits naturally into the squared-exponential
approximation framework with non-negative coefficients. 

\subsection{Numerical estimation of squared-exponential approximations}\label{sec:numerical_examples_square_approx}

Let us focus on problems where the measure $\mu$ in~\eqref{eq:cm_representation}
has infinite support. In principle, the exponents $b_{i}$ could be
chosen as the abscissae of a Gaussian quadrature rule for~$\mu$.
However, the main challenge lies in the fact that the measure itself
is \emph{unknown}. We therefore adopt an optimization strategy to
determine an optimal set of $b_{i}$ for a given approximation rank
$k$. 

\subsubsection{Problem setting}\label{subsec:Problem-setting-kernels}

We approximate a one-dimensional covariance kernel by a non-negative
sum of $k$ squared-exponential terms 
\begin{equation}
C(d)\;\approx\;\sum_{i=1}^{k}a_{i}\mathrm{e}^{-b_{i}d^{2}},\qquad a_{i}>0,\;b_{i}>0,\label{eq:se-expansion}
\end{equation}
where $k\in\mathbb{N}$ is the \emph{rank} of the approximation. 

We consider the following five benchmark covariance kernels, all normalized
to a unit length scale and defined over the domain $d\in[0,2]$: 
\begin{itemize}
\item Exponential kernel: $\exp(-d)$ \cite{Rasmussen2006}; 
\item Matérn kernel ($\nu=\tfrac{5}{2}$): $C_{\frac{5}{2}}(d)=\bigl(1+\sqrt{5}d+\tfrac{5}{3}d^{2}\bigr)\exp(-\sqrt{5}d)$
\cite{Stein1999}; 
\item Stretched exponential kernel: $\exp\bigl(-d^{0.6}\bigr)$ \cite{Beran1994}; 
\item Rational--quadratic kernel: $\bigl(1+d^{2}/2\bigr)^{-1}$ \cite{Rasmussen2006}; 
\item Cauchy kernel: $\dfrac{1}{1+d}$, a special case of the (generalized)
Cauchy family \cite{Gneiting2004}. 
\end{itemize}
This set spans exponential, algebraic, and heavy-tailed behaviors,
posing various challenges for the finite-rank representation~\eqref{eq:se-expansion}.
We retain the domain $d\in[0,2]$ because all pairwise distances that
arise in subsequent experiments fall within this range. These experiments
involve decomposing covariance functions on hypercubes $[0,1]^{D}$,
where $D=1,2,3$. The maximum diameter of these domains is $\sqrt{3}$,
which ensures that all relevant distances are encompassed by the chosen
interval. 

\subsubsection{Optimization problem}

Before formulating the optimization problem, we introduce a reparameterization
that ensures the coefficients $b_{i}=\mathrm{e}^{\theta_{i}}$ remain
positive. In other words, we will optimize over the parameters $\theta_{i}$
instead. This reparameterization also addresses scaling issues, as
the values of $b_{i}$ grow exponentially, as demonstrated in subsequent
experiments. 

The continuous $L^{2}$ error is 
\[
\mathcal{J}(\mathbf{a},\boldsymbol{\theta})=\left\Vert \sum_{i=1}^{k}a_{i}\mathrm{e}^{-\mathrm{e}^{\theta_{i}}d^{2}}-f(d)\right\Vert _{2}=\Bigl[\int_{0}^{2}\Bigl(\sum_{i=1}^{k}a_{i}\mathrm{e}^{-\mathrm{e}^{\theta_{i}}d^{2}}-f(d)\Bigr)^{2}\mathrm{d}d\Bigr]^{1/2},\,\,b_{i}=\mathrm{e}^{\theta_{i}}.
\]
For computational practicality, the integral is approximated via numerical
integration. In this context, we propose a \emph{multi-level} Gauss--Legendre
rule, i.e., a composite rule consisting of Gauss--Legendre quadrature
applied over multiple intervals. The multi-level structure partitions
the interval $[0,L]$ into a sequence of subintervals: $[0,L\rho^{m}],[L\rho^{m},L\rho^{m-1}],\ldots,[L\rho^{1},L]$,
where $m\in\mathbb{N}$ denotes the number of levels and $\rho\in(0,1)$
is the division ratio. This form of numerical quadrature is specifically
designed to capture critical behavior near zero, particularly when
$b_{i}$ is large, as the squared-exponential kernel then requires
fine resolution near zero for accurate approximation. Numerical experiments
were conducted with $m=5$, geometric ratio $\rho=0.2$, and quadrature
degree $n=100$ on each subinterval, resulting in $(L+1)n=600$ nodes
$(x_{j},w_{j})_{j=1}^{600}$. The corresponding discrete objective
function is given by 
\[
J(\mathbf{a},\boldsymbol{\theta})=\Bigl[\sum_{j=1}^{600}w_{j}\Bigl(\sum_{i=1}^{k}a_{i}\mathrm{e}^{-\mathrm{e}^{\theta_{i}}x_{j}^{2}}-f(x_{j})\Bigr)^{2}\Bigr]^{1/2}.
\]

Note that, due to the structure of the problem, the objective functionals
$\mathcal{J}(\mathbf{a},\boldsymbol{\theta})$ and $J(\mathbf{a},\boldsymbol{\theta})$
possess multiple global minima, as the individual $(b_{i},a_{i})$
pairs can be interchanged without affecting the value of the functional. 

\subsubsection{Optimization Algorithm}

\label{subsec:Optimisation-algorithm}

The unknowns are $\mathbf{p}=(\mathbf{a},\boldsymbol{\theta})\in\mathbb{R}^{2k}$,
with non-negative $a_{i}\ge0$ and unconstrained $\theta_{i}\in\mathbb{R}$.
The non-negativity constraint is not enforced in the Newton algorithm,
as the global minima are expected to naturally yield positive values.
The minimization process is summarized as follows:
\begin{description}
\item [{Newton--type~iteration}] At iteration $\mathbf{p}_{j}$, we compute
the gradient $\nabla J_{j}$ and the Hessian $H_{j}$ using \textsc{JAX}
automatic differentiation (64-bit precision is necessary for stability).
A damped Newton search direction $\mathbf{s}_{j}$ is obtained from
\[
\bigl(H_{j}+\tau_{j}I\bigr)\mathbf{s}_{j}=-\nabla J_{j},\qquad\tau_{j}\ge0.
\]

\begin{enumerate}
\item \emph{Positive-definite case:} If $H_{j}+\tau_{j}I$ is positive definite,
i.e., $\nabla J_{j}^{\top}\mathbf{s}_{j}<0$ so that $\mathbf{s}_{j}$
is a descent direction, set $\mathbf{d}_{j}=\mathbf{s}_{j}$ and decrease
the damping parameter ($0\leq\tau_{j+1}<\tau_{j}$).
\item \emph{Indefinite case:} Otherwise, increase the damping ($\tau_{j+1}>\tau_{j}$)
and fall back to steepest descent, $\mathbf{d}_{j}=-\nabla J_{j}$.
\end{enumerate}
\item [{Line~search}] \noindent A backtracking line search on $[0,1]$,
using the golden-section ratio, determines the step length $\alpha_{j}\in(0,1]$.
\item [{Update}] \noindent
\[
\mathbf{p}_{j+1}=\mathbf{p}_{j}+\alpha_{j}\mathbf{d}_{j}.
\]
\item [{Stopping~criteria}] \noindent Terminate when
\begin{enumerate}
\item $\|\nabla J_{j}\|_{2}<\varepsilon_{\nabla}$; 
\item $\|\mathbf{p}_{j+1}-\mathbf{p}_{j}\|_{2}<\varepsilon_{\mathbf{p}}$
for the last $n_{\mathbf{p}}$ iterations; or 
\item the maximum iteration count is exceeded. 
\end{enumerate}
\end{description}

\subsubsection{Initialization Strategy}

\label{sec:init-strategy}

The minimization of $J(\mathbf{a},\boldsymbol{\theta})$ via Newton’s
method is reliable only when the initial guess is sufficiently close
to the minimizer; if it is too far, the iteration may diverge. Hence,
one must employ either (i)~a global search strategy or (ii)~a \emph{good}
initial guess. We adopt the second approach because, thanks to the
structure of the problem, such a guess can be reasonably estimated.

For a fixed rank~$k$, the squared-exponential sum in \eqref{eq:se-expansion}
may be interpreted as a $k$-point quadrature rule for the Laplace--Stieltjes
integral in \eqref{eq:cm_representation}, with \emph{nodes}~$b_{i}$
and \emph{weights}~$a_{i}$. Gaussian quadrature minimizes the integration
error among all rules with $k$ nodes; therefore, the optimization
problem is expected to recover the same node set. Denote the optimal
exponents by $^{\star}b_{i}^{k}$, $i=1,\dots,k$. Because the Gaussian
nodes of consecutive orders alternate \cite[Thm. 1.20]{Gautschi2004},
the optimal exponents are expected to satisfy 
\[
0<^{\star}b_{1}^{k+1}<{}^{\star}b_{1}^{k}<{}^{\star}b_{2}^{k+1}<{}^{\star}b_{2}^{k}<\cdots<{}^{\star}b_{k}^{k}<{}^{\star}b_{k+1}^{k+1},\qquad k\ge1.
\]
We exploit this property and additionally work with the log-transformed
parameters $\theta_{i}=\log b_{i}$ to accommodate the exponential
growth of~$b_{i}$. Initial guesses $^{(0)}\theta_{i}^{k}$ that
preserve the observed interior gap ratios are constructed as detailed
in Algorithm~\ref{alg:init-theta}. Once the $\theta_{i}$ have been
initialized, the corresponding coefficients $a_{i}$ are obtained
via non-negative least squares (NNLS). 

\begin{algorithm}[t]
\caption{Initial guess $^{(0)}\theta_{i}^{k}$ for rank~$k$}
\label{alg:init-theta} \begin{algorithmic}[1] \Require maximum
rank $k_{\max}$ 

\Statex\textbf{Base ranks} 

\State $^{(0)}\theta_{1}^{1}\gets0$

\State $^{(0)}\theta_{1}^{2}\gets{}^{\star}\theta_{1}^{1}-1$,\;
$^{(0)}\theta_{2}^{2}\gets{}^{\star}\theta_{1}^{1}+1$ 

\Statex\textbf{Ranks $k\ge3$} 

\For{$k=3$ \textbf{to} $k_{\max}$} 

\State $^{(0)}\theta_{1}^{k}\gets{}^{\star}\theta_{1}^{k-1}+\bigl({}^{\star}\theta_{1}^{k-1}-{}^{\star}\theta_{1}^{k-2}\bigr)$ 

\State $^{(0)}\theta_{k}^{k}\gets{}^{\star}\theta_{k-1}^{k-1}+\bigl({}^{\star}\theta_{k-1}^{k-1}-{}^{\star}\theta_{k-2}^{k-2}\bigr)$ 

\For{$i=2$ \textbf{to} $k-1$} 

\State $r_{i}^{k}\gets\begin{cases}
\tfrac{1}{2}, & i=k-1,k=3,\\
r_{i-1}^{k}, & i=k-1,k>3\\
r_{i}^{k}:\,{}^{\star}\theta_{i}^{k-1}=r_{i}^{k}{}^{\star}\theta_{i-1}^{k-2}+(1-r_{i}^{k}){}^{\star}\theta_{i}^{k-2}, & \text{otherwise}
\end{cases}$ 

\State $^{(0)}\theta_{i}^{k}\gets r_{i}^{k}\,{}^{\star}\theta_{i-1}^{k-1}+(1-r_{i}^{k})\,{}^{\star}\theta_{i}^{k-1}$ 

\EndFor 

\EndFor 

\end{algorithmic} 
\end{algorithm}

\subsubsection{Experiments}

\begin{figure}[t]
\centering{}\includegraphics{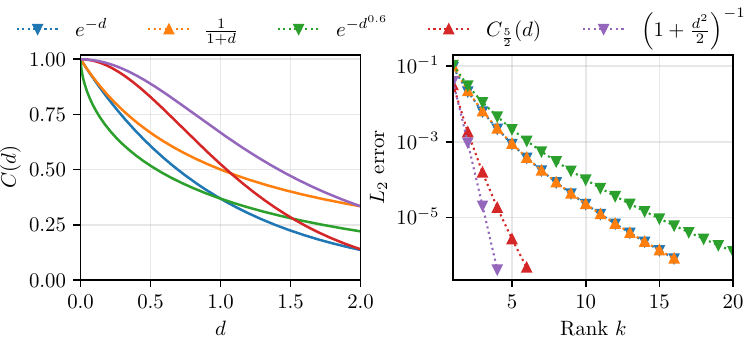} \caption{Left: Illustration of selected target kernels. Right: Approximation
of the five target kernels by squared-exponential approximation with
increasing rank size until $J(\mathbf{a},\boldsymbol{\theta})<10^{-6}$.}
\label{fig:cov-approx}
\end{figure}

\begin{figure}[t]
\centering{}\includegraphics{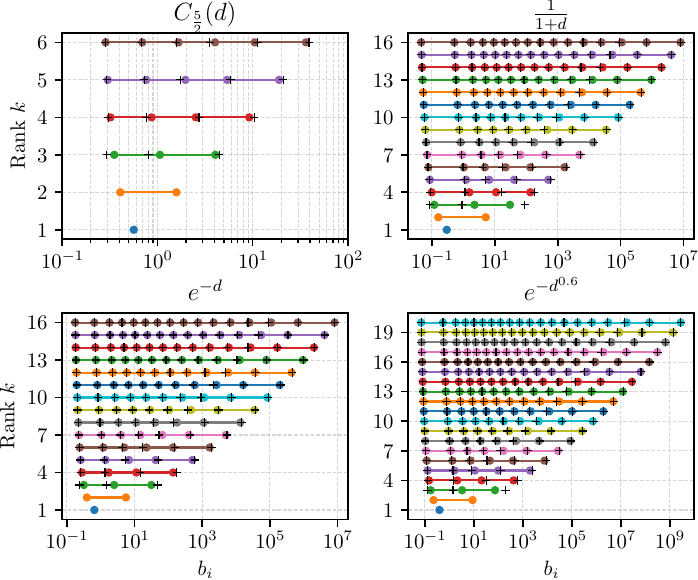} \caption{Optimal coefficients $b_{i}$ for computed ranks for four selected
kernels (black + signs illustrate the initial guesses).}
\label{fig:b-coeffs}
\end{figure}

\begin{table}[t]
\centering{}%
\begin{tabular}{lcrcrcrcr}
\toprule 
\multirow{1}{*}{kernel} & \multicolumn{2}{c}{$\varepsilon=10^{-2}$} & \multicolumn{2}{c}{$\varepsilon=10^{-3}$} & \multicolumn{2}{c}{$\varepsilon=10^{-4}$} & \multicolumn{2}{c}{$\varepsilon=10^{-6}$}\tabularnewline
\midrule 
 & $k$ & exact~$\varepsilon$ & $k$ & exact~$\varepsilon$ & $k$ & exact~$\varepsilon$ & $k$ & exact~$\varepsilon$\tabularnewline
\midrule 
$\mathrm{e}^{-d}$ & 3 & $6\times10^{-3}$ & 5 & $9\times10^{-4}$ & 8 & $9\times10^{-5}$ & 16 & $8\times10^{-7}$\tabularnewline
$C_{5/2}(d)$ & 2 & $2\times10^{-3}$ & 3 & $2\times10^{-4}$ & 4 & $2\times10^{-5}$ & 6 & $5\times10^{-7}$\tabularnewline
$\mathrm{e}^{-d^{0.6}}$ & 4 & $5\times10^{-3}$ & 7 & $6\times10^{-4}$ & 11 & $6\times10^{-5}$ & 20 & $1\times10^{-6}$\tabularnewline
$(1+d^{2}/2)^{-1}$ & 2 & $1\times10^{-3}$ & 2 & $1\times10^{-3}$ & 3 & $2\times10^{-5}$ & 4 & $4\times10^{-7}$\tabularnewline
$1/(1+d)$ & 3 & $6\times10^{-3}$ & 5 & $9\times10^{-4}$ & 8 & $9\times10^{-5}$ & 16 & $8\times10^{-7}$\tabularnewline
\bottomrule
\end{tabular}\caption{Ranks required to reach four target accuracies. Entries labeled "exact~$\varepsilon$"
show the actual achieved error for the reported rank.}
\label{tab:ranks}
\end{table}

We combine the optimization algorithm of Section~\ref{subsec:Optimisation-algorithm}
with the initialization strategy of Section~\ref{sec:init-strategy}.
Starting at rank~$k=1$, we increase the rank successively, with
each converged solution providing the initial guess for the next run.
The sequence terminates as soon as $J<10^{-6}$; otherwise, it proceeds
up to the upper limit $k_{\max}=20$. All calculations are performed
in double precision, which is required for the stable automatic differentiation
used in computing the Hessians.

Figure~\ref{fig:cov-approx} displays the five target kernels alongside
their squared-exponential approximations for representative ranks.
The plots reveal the expected exponential error decay, with a base
rate that depends on each kernel’s rate of decay near $d=0$. For
subsequent demonstrations, we refer to an approximation by the tolerance
$\varepsilon$ that bounds its discrete $L^{2}$ error. The ranks
required to achieve four tolerance levels, along with the corresponding
errors, are listed in Table~\ref{tab:ranks}. Figure~\ref{fig:b-coeffs}
shows the optimized exponents~$b_{i}$ alongside their initial guesses;
the alternating property, approximate exponential growth, and the
quality of the initial guesses are clearly evident.

\section{Numerical approximation of Karhunen--Loève expansion}\label{sec:Numerical-approximation-of}

Finding the exact eigenfunctions and eigenvalues of the autocovariance
operator is, in most cases, impossible. Except for specific cases
(e.g., exponential covariance on a one-dimensional interval), the
eigenvalue problem \eqref{eq:kl_eigenproblem} rarely admits an analytical
solution. Therefore, numerical approximation methods are essential.
We explore the use of Galerkin projection onto finite-dimensional
subspaces spanned by orthonormal polynomials. Specifically, we consider
tensor products of one-dimensional orthonormal polynomials, which
are well-suited for hyper-rectangular domains. 

\subsection{Well-Posedness of the Spectral Eigenvalue Problem}

The existence and properties of the KL expansion depend crucially
on the spectral characteristics of the autocovariance operator~$\mathcal{C}$
defined in \eqref{eq:cov_operator}. It is necessary to ensure that
the eigenvalue problem \eqref{eq:kl_eigenproblem} is well-posed,
i.e., that it admits a countable set of non-negative eigenvalues and
corresponding orthonormal eigenfunctions. 
\begin{prop}
\label{prop:operator_properties} Let $C(\bm{x},\bm{y})$ be the autocovariance
function of a zero-mean, second-order random field~$Z$ on the bounded
domain~$\Omega$, and assume that $C$ is continuous on~$\Omega\times\Omega$.
Then the integral operator $\mathcal{C}:L^{2}(\Omega)\to L^{2}(\Omega)$,
defined by \eqref{eq:cov_operator}, is linear, bounded, self-adjoint,
positive semi-definite, and compact. 
\end{prop}

\begin{proof}
Linearity is immediate. Boundedness follows from the continuity of
$C$ on the compact set~$\Omega\times\Omega$. 
\begin{itemize}
\item \textbf{Self-adjointness:} Since $C(\bm{x},\bm{y})=\mathbb{E}[Z(\bm{x})Z(\bm{y})]=\mathbb{E}[Z(\bm{y})Z(\bm{x})]=C(\bm{y},\bm{x})$,
the kernel is symmetric. For any $f,g\in L^{2}(\Omega)$: 
\begin{align*}
\langle\mathcal{C}f,g\rangle_{L^{2}(\Omega)} & =\langle f,\mathcal{C}g\rangle_{L^{2}(\Omega)}.
\end{align*}
\item \textbf{Positive semi-definiteness:} For any $f\in L^{2}(\Omega)$:
\begin{align*}
\langle\mathcal{C}f,f\rangle_{L^{2}(\Omega)} & =\int_{\Omega}\int_{\Omega}\mathbb{E}[Z(\bm{x})Z(\bm{y})]f(\bm{y})f(\bm{x})\,d\bm{y}\,d\bm{x}\\
 & =\mathbb{E}\left[\left(\int_{\Omega}Z(\bm{x})f(\bm{x})\,d\bm{x}\right)^{2}\right]\ge0.
\end{align*}
\item \textbf{Compactness:} Since $\Omega$ is a bounded domain and $C$
is continuous on~$\Omega\times\Omega$, it is square-integrable:
$\int_{\Omega}\int_{\Omega}|C(\bm{x},\bm{y})|^{2}\,d\bm{x}\,d\bm{y}<\infty$.
Operators defined by such kernels (Hilbert--Schmidt integral operators)
are compact on $L^{2}(\Omega)$ \cite[Prop.~4.7]{Conway1990}. 
\end{itemize}
\end{proof}
The properties established in Proposition~\ref{prop:operator_properties}
allow us to invoke the Spectral Theorem for compact, self-adjoint
operators on Hilbert spaces. 
\begin{thm}[Spectral Theorem, Hilbert-Schmidt version \cite{ReedSimon1980}, \cite{Kreyszig1989}]
\label{thm:spectral} Let $\mathcal{T}:H\to H$ be a compact, self-adjoint
operator on a Hilbert space~$H$. Then there exists a sequence of
real eigenvalues~$(\lambda_{k})_{k=1}^{\infty}$ and an orthonormal
sequence of corresponding eigenfunctions~$(\phi_{k})_{k=1}^{\infty}$
in~$H$ such that $\mathcal{T}\phi_{k}=\lambda_{k}\phi_{k}$. The
sequence of eigenvalues converges to zero ($\lambda_{k}\to0$ as $k\to\infty$),
and the eigenfunctions form an orthonormal basis for $(\ker\mathcal{T})^{\perp}=\overline{\text{range }\mathcal{T}}$.
Any $f\in H$ can be represented as 
\[
f=\sum_{k=1}^{\infty}\langle f,\phi_{k}\rangle\phi_{k}+f_{0},
\]
where $f_{0}\in\ker\mathcal{T}$, and 
\[
\mathcal{T}f=\sum_{k=1}^{\infty}\lambda_{k}\langle f,\phi_{k}\rangle\phi_{k}.
\]
\end{thm}

Applying Theorem~\ref{thm:spectral} to the autocovariance operator~$\mathcal{C}$,
which is compact, self-adjoint, and positive semi-definite (Proposition~\ref{prop:operator_properties}),
guarantees the existence of a countable sequence of non-negative eigenvalues~$\lambda_{k}\ge0$
with $\lambda_{k}\to0$, and corresponding orthonormal eigenfunctions
$\phi_{k}\in L^{2}(\Omega)$. This confirms that the eigenvalue problem
\eqref{eq:kl_eigenproblem}, which defines the KL basis, is well-posed. 

\subsection{Tensor-product polynomial Galerkin approximation}\label{sec:galerkin}

As mentioned earlier, we work on the $D$‑dimensional hyper‑rectangle
\[
\Omega=\prod_{l=1}^{D}(a_{l},b_{l})\subset\mathbb{R}^{D},
\]
as it permits the use of tensor‑product orthonormal Legendre polynomials.
For each coordinate $l=1,\dots,D$, let 
\[
\{\varphi_{\alpha}^{(l)}\}_{\alpha\ge0}\subset L^{2}\bigl((a_{l},b_{l})\bigr)
\]
denote the Legendre polynomials that are \emph{orthonormal} on the
interval~$(a_{l},b_{l})$. For a multi-index $\boldsymbol{\alpha}=(\alpha_{1},\dots,\alpha_{D})\in\mathbb{N}_{0}^{D}$
with $0\le\alpha_{l}\le n$, we define the tensor‑product basis functions
\[
\psi_{\boldsymbol{\alpha}}(\bm{x}):=\prod_{l=1}^{D}\varphi_{\alpha_{l}}^{(l)}(x_{l}),\qquad\bm{x}=(x_{1},\dots,x_{D})\in\Omega.
\]
Since each factor is orthonormal in its own coordinate, the set $\{\psi_{\boldsymbol{\alpha}}\}$
forms an orthonormal basis of 
\[
V_{n}:=\operatorname{span}\bigl\{\psi_{\boldsymbol{\alpha}}:0\le\alpha_{l}\le n\bigr\}\subset L^{2}(\Omega),\qquad N_{n}:=\dim V_{n}=(n+1)^{D}.
\]

Let $\mathcal{C}:L^{2}(\Omega)\to L^{2}(\Omega)$ be the covariance
operator of the random field under consideration. Its Galerkin projection
$\mathcal{C}_{n}:V_{n}\to V_{n}$ is defined by 
\begin{equation}
\bigl(\mathcal{C}_{n}u_{n},v_{n}\bigr)_{L^{2}(\Omega)}:=\bigl(\mathcal{C}u_{n},v_{n}\bigr)_{L^{2}(\Omega)},\qquad u_{n},v_{n}\in V_{n}.\label{eq:galerkinCn}
\end{equation}
With respect to the basis~$\{\psi_{\boldsymbol{\alpha}}\}$, this
operator is represented by the symmetric, positive semi-definite matrix
\[
A_{n}:=\bigl[(\mathcal{C}\psi_{\boldsymbol{\beta}},\psi_{\boldsymbol{\alpha}})_{L^{2}(\Omega)}\bigr]_{\boldsymbol{\alpha},\boldsymbol{\beta}}\in\mathbb{R}^{N_{n}\times N_{n}}.
\]

The Galerkin approximation of the KL expansion is then defined by
combining \eqref{eq:kl_eigenproblem} and \eqref{eq:galerkinCn}:
\[
\bigl(\mathcal{C}\phi_{j}^{(n)},v_{n}\bigr)_{L^{2}(\Omega)}=\lambda_{j}^{(n)}\bigl(\phi_{j}^{(n)},v_{n}\bigr)_{L^{2}(\Omega)},\quad\forall v_{n}\in V_{n},\;j=1,\dots,N_{n}.
\]
With respect to the basis~$\{\psi_{\boldsymbol{\alpha}}\}$, this
becomes the generalized eigenvalue problem $A_{n}\mathbf{u}_{j}^{(n)}=\lambda_{j}^{(n)}M_{n}\mathbf{u}_{j}^{(n)}$.
Note that the mass matrix on the right-hand side is the identity,
since the basis~$\{\psi_{\boldsymbol{\alpha}}\}$ is \emph{orthonormal}
in $L^{2}(\Omega)$. Therefore, the problem reduces to the matrix
eigenvalue problem 
\begin{equation}
A_{n}\mathbf{u}_{j}^{(n)}=\lambda_{j}^{(n)}\mathbf{u}_{j}^{(n)},\qquad\mathbf{u}_{j}^{(n)}\in\mathbb{R}^{N_{n}},\;\lambda_{j}^{(n)}\ge0.\label{eq:disc_KL}
\end{equation}

\begin{thm}
\label{thm:galerkin_wellposed} For every $n\ge0$, the problem \eqref{eq:disc_KL}
admits $N_{n}$ real, non-negative eigenvalues 
\[
0\le\lambda_{N_{n}}^{(n)}\le\cdots\le\lambda_{1}^{(n)},
\]
with corresponding $L^{2}(\Omega)$-orthonormal eigenfunctions 
\[
\phi_{j}^{(n)}:=\sum_{\boldsymbol{\alpha}}u_{j,\boldsymbol{\alpha}}^{(n)}\psi_{\boldsymbol{\alpha}}\in V_{n},\qquad j=1,\dots,N_{n}.
\]
Hence, $\mathcal{C}_{n}$ is diagonalizable, and the discrete KL expansion
\[
X_{n}=\sum_{j=1}^{N_{n}}\sqrt{\lambda_{j}^{(n)}}\xi_{j}\phi_{j}^{(n)},\qquad\xi_{j}\stackrel{\text{i.i.d.}}{\sim}\mathcal{N}(0,1)
\]
is well-defined and forms a Gaussian random field with covariance
kernel 
\[
C_{n}(\bm{x},\bm{y}):=\sum_{j=1}^{N_{n}}\lambda_{j}^{(n)}\phi_{j}^{(n)}(\bm{x})\phi_{j}^{(n)}(\bm{y}),\qquad\bm{x},\bm{y}\in\Omega.
\]
\end{thm}

\subsection{Tensor-structured assembly of the Galerkin matrices }\label{subsec:tensor-assembly}

The squared-exponential representation derived in Section~\ref{sec:numerical_examples_square_approx}
yields a separable approximation of the isotropic covariance kernel:
\[
\mathcal{C}(\bm{x},\bm{y})\;\approx\;C_{\mathrm{sep}}(\bm{x},\bm{y}):=\sum_{i=1}^{k}a_{i}\prod_{l=1}^{D}\exp\bigl(-b_{i}(x_{l}-y_{l})^{2}\bigr),\qquad a_{i}>0,\;b_{i}>0,
\]
where $D$ is the spatial dimension and $k$ the separation rank.
Together with the tensor-product Legendre basis~$\{\psi_{\boldsymbol{\alpha}}\}_{0\le\alpha_{l}\le n}$
introduced in Section~\ref{sec:galerkin}, the entries of the Galerkin
matrix $A_{n}\in\mathbb{R}^{N_{n}\times N_{n}}$ (with $N_{n}=(n+1)^{D}$)
are given by 
\[
(A_{n})_{\boldsymbol{\beta},\boldsymbol{\alpha}}=\iint_{\Omega\times\Omega}C_{\mathrm{sep}}(\bm{x},\bm{y})\psi_{\boldsymbol{\alpha}}(\bm{y})\psi_{\boldsymbol{\beta}}(\bm{x})\;\mathrm{d}\bm{y}\,\mathrm{d}\bm{x}.
\]

Because both the kernel and the basis are fully separable, each double
integral factorizes into a product of one-dimensional integrals. Writing
$\boldsymbol{\alpha}=(\alpha_{1},\dots,\alpha_{D})$ and $\boldsymbol{\beta}=(\beta_{1},\dots,\beta_{D})$,
we obtain 
\begin{equation}
(A_{n})_{\boldsymbol{\beta},\boldsymbol{\alpha}}=\sum_{i=1}^{k}a_{i}\prod_{l=1}^{D}\underbrace{\int_{a_{l}}^{b_{l}}\int_{a_{l}}^{b_{l}}e^{-b_{i}(x-y)^{2}}\varphi_{\alpha_{l}}^{(l)}(y)\varphi_{\beta_{l}}^{(l)}(x)\,\mathrm{d}y\,\mathrm{d}x}_{=:\bigl(\mathbf{A}_{n,i}^{(l)}\bigr)_{\beta_{l}\alpha_{l}}},\label{eq:sub_matrix}
\end{equation}
where $\mathbf{A}_{n,i}^{(l)}\in\mathbb{R}^{(n+1)\times(n+1)}$ depends
only on the polynomial degree~$n$, the rank index~$i$, and the
spatial direction~$l$. Consequently, 
\begin{equation}
A_{n}=\sum_{i=1}^{k}a_{i}\;\bigotimes_{l=1}^{D}\mathbf{A}_{n,i}^{(l)}.\label{eq:tensor_matrix}
\end{equation}

\subsubsection{Computational advantages}

The use of the squared-exponential approximation offers two main computational
advantages: 
\begin{itemize}
\item \emph{Quadrature effort.} Direct assembly of $A_{n}$ would require
$N_{n}^{2}=(n+1)^{2D}$ integrals over $\mathbb{R}^{2D}$. Formula~\eqref{eq:tensor_matrix}
reduces this to $kD(n+1)^{2}$ integrals over \emph{two-dimensional}
domains, representing a dramatic saving even for moderate values of~$D$. 
\item \emph{Matrix--vector products.} Naively applying $A_{n}$ to a vector
of length~$N_{n}$ incurs a cost of $\mathcal{O}(N_{n}^{2})=\mathcal{O}((n+1)^{2D})$
operations. By exploiting the tensor structure, this cost is reduced
to $\mathcal{O}\bigl(kD(n+1)^{D+1}\bigr)$. 
\end{itemize}
Consider a realistic setting with polynomial degree $n=100$ in $D=2$
dimensions and a tensor sum with $k=16$ terms. The naive cost of
a matrix--vector multiplication grows as $\mathcal{O}((n+1)^{2D})$,
which evaluates to about $101^{4}\approx10^{8}$ operations. In contrast,
the cost of the tensor-structured approach scales as $\mathcal{O}(kD(n+1)^{D+1})$,
which corresponds to approximately $16\cdot2\cdot101^{3}\approx3.3\times10^{7}$
operations. This results in roughly 3x speedup of matrix--vector
multiplication.

The difference in integration cost is even more dramatic. Without
separability, assembling $A_{n}$ would require $N_{n}^{2}$ integrals
over $2D$ dimensions. For $n=100$ and $D=2$, this corresponds to
approximately $10^{8}$ integrals over~$\mathbb{R}^{4}$. Using a
tensor-product quadrature grid with 400 points per dimension, this
would demand about $400^{4}=2.56\times10^{10}$ function evaluations
per integral, totaling over $10^{18}$ evaluations. In contrast, the
separated form involves only $kD(n+1)^{2}\approx3.3\times10^{5}$
integrals over~$\mathbb{R}^{2}$, each requiring only $400^{2}=1.6\times10^{5}$
function evaluations, for a total cost of approximately $5.2\times10^{10}$
evaluations. 

\subsubsection{Parity structure and block decomposition}

\label{subsec:parity-blocks} 

Additional speed-ups in the assembly, matrix--vector products, and
final eigensolve are achieved by exploiting the \emph{parity} of the
Legendre basis functions appearing in the one-dimensional sub-matrices
$\mathbf{A}_{n,i}^{(l)}$.

After shifting each spatial interval to be centered at the origin
$(-\gamma,\gamma)$, the Legendre polynomials satisfy 
\[
\varphi_{\alpha}(-x)=(-1)^{\alpha}\varphi_{\alpha}(x),\qquad\alpha=0,1,\dots,n.
\]
Thus, $\varphi_{\alpha}$ is \emph{even} when $\alpha$ is even and
\emph{odd} when $\alpha$ is odd.

For each kernel term in the separated covariance approximation, the
one-dimensional factors introduced in~\eqref{eq:sub_matrix} are
given by 
\[
\bigl(\mathbf{A}_{n,i}^{(l)}\bigr)_{\alpha\beta}=\int_{-\gamma}^{\gamma}\int_{-\gamma}^{\gamma}\exp\bigl(-b_{i}(x-y)^{2}\bigr)\varphi_{\alpha}(y)\varphi_{\beta}(x)\mathrm{d}y\mathrm{d}x,\qquad0\le\alpha,\beta\le n,\;i=1,\dots,k.
\]

\begin{lem}[parity orthogonality \cite{Beres2018,Beres2023}]
\label{lem:parity} Let $p:[0,\infty)\to\mathbb{R}$, and let $\varphi_{\text{odd}}$
and $\varphi_{\text{even}}$ be odd and even functions, respectively,
on $(-\gamma,\gamma)$. Then 
\[
\int_{-\gamma}^{\gamma}\int_{-\gamma}^{\gamma}p\bigl(|x-y|\bigr)\varphi_{\text{odd}}(x)\varphi_{\text{even}}(y)\mathrm{d}y\mathrm{d}x=0.
\]
\end{lem}

Taking $p(t)=\exp(-b_{i}t^{2})$ in Lemma~\ref{lem:parity} yields
\[
(\mathbf{A}_{n,i}^{(l)})_{\alpha\beta}=0\quad\Longleftrightarrow\quad\alpha+\beta\ \text{is odd}.
\]
Permuting the basis so that the even degrees precede the odd degrees
partitions each $\mathbf{A}_{n,i}^{(l)}$ into two diagonal blocks:
\[
\mathbf{A}_{n,i}^{(l)}\;=\;\begin{bmatrix}\mathbf{A}_{n,i,\mathrm{even}}^{(l)} & \mathbf{0}\\
\mathbf{0} & \mathbf{A}_{n,i,\mathrm{odd}}^{(l)}
\end{bmatrix},\qquad i=1,\dots,k,\;l=1,\dots,D,
\]
with $\mathbf{A}_{n,i,\mathrm{even}}^{(l)}\in\mathbb{R}^{n_{e}\times n_{e}}$,
$\mathbf{A}_{n,i,\mathrm{odd}}^{(l)}\in\mathbb{R}^{n_{o}\times n_{o}}$,
where $n_{e}=\lceil(n+1)/2\rceil$ and $n_{o}=\lfloor(n+1)/2\rfloor$. 

Let 
\[
\epsilon=(\epsilon_{1},\dots,\epsilon_{D})\in\{0,1\}^{D},\qquad\epsilon_{l}=0\;\text{for even,}\;\epsilon_{l}=1\;\text{for odd parity}.
\]
Define the parity-specific factors 
\[
\mathbf{A}_{n,i}^{(l,\epsilon_{l})}:=\begin{cases}
\mathbf{A}_{n,i,\mathrm{even}}^{(l)}, & \epsilon_{l}=0,\\
\mathbf{A}_{n,i,\mathrm{odd}}^{(l)}, & \epsilon_{l}=1.
\end{cases}
\]
After a global even/odd permutation $P$, the Galerkin matrix becomes
block diagonal: 
\[
\widetilde{A}_{n}:=P^{\top}A_{n}P=\bigoplus_{\epsilon\in\left\{ 0,1\right\} ^{D}}A_{n}^{\epsilon},\qquad A_{n}^{\epsilon}=\sum_{i=1}^{k}a_{i}\bigotimes_{l=1}^{D}\mathbf{A}_{n,i}^{(l,\epsilon_{l})}.
\]
Each block $A_{n}^{\epsilon}$ acts only on the tensor-product subspace
spanned by basis functions with the prescribed parity vector $\epsilon$.
Its dimension is given by 
\[
N_{\epsilon}=\prod_{l=1}^{D}\bigl(n_{e}^{1-\epsilon_{l}}n_{o}^{\epsilon_{l}}\bigr)\;\le\;\max\{n_{e},n_{o}\}^{D}\approx\bigl(\tfrac{n+1}{2}\bigr)^{D},
\]
whereas the full space has dimension $N_{n}=(n+1)^{D}$. Thus, each
block is smaller by a factor of approximately $2^{D}$. 

Because the blocks $A_{n}^{\epsilon}$ are mutually independent, the
eigenproblems $A_{n}^{\epsilon}\mathbf{u}=\lambda\mathbf{u}$ can
be solved separately and fully in parallel. One iteration of an Arnoldi-type
routine, such as \texttt{eigs}, requires a single matrix--vector
multiplication. For the \emph{full} tensor matrix, this costs 
\[
\mathcal{O}\bigl(kD(n+1)^{D+1}\bigr)\qquad\text{(see Section~\ref{subsec:tensor-assembly}).}
\]
After the parity permutation, the largest block has dimension $\smash{\bigl(\tfrac{n+1}{2}\bigr)^{D}}$;
consequently, 
\[
\mathcal{O}\bigl(kD((n+1)/2)^{D+1}\bigr)=\frac{1}{2^{D+1}}\mathcal{O}\bigl(kD(n+1)^{D+1}\bigr),
\]
i.e., each block--matvec is $2^{D+1}$ times cheaper than the corresponding
operation with the unsplit matrix, while there are only $2^{D}$ such
blocks to process. 

When the goal is to compute the leading $M2^{D}$ eigenpairs, one
now computes $M$ eigenpairs on each of the $2^{D}$ blocks, instead
of computing $M2^{D}$ eigenpairs for the full matrix. Because Krylov
and Arnoldi algorithms grow more than linearly in cost with the number
of requested eigenpairs, this block-wise approach offers an additional
practical acceleration beyond the formal factor of $2$. 

\subsection{Stable quadrature for the one--dimensional blocks }\label{subsec:duffy-quadrature}

The accurate evaluation of the matrices $\mathbf{A}_{n,i}^{(l)}$
in~(\ref{eq:sub_matrix}) is the final component that remains to
be addressed. Although one can symbolically differentiate the Legendre
polynomials and integrate them in closed form, the term $\exp\bigl(-b_{i}(x-y)^{2}\bigr)$
renders the computation numerically unstable, even for low polynomial
orders, when $b_{i}\gg1$. Instead, we rely on numerical quadrature,
preceded by a coordinate transformation that broadens the narrow peak
of the kernel. We demonstrate the approach on the reference square
$[0,1]^{2}$; however, generalization to any square is straightforward---for
example, via an initial substitution into $[0,1]^{2}$, which simply
modifies $b_{i}$ and introduces a multiplicative factor in the integral. 

\subsubsection{Duffy split and algebraic stretching}

On the reference square $[0,1]^{2}$, the kernel decays from $1$
to $\mathrm{e}^{-1}$ once $|x-y|\gtrsim b_{i}^{-1/2}$. For large
$b_{i}$, the informative part is confined to a thin strip around
the diagonal $x=y$, so a naive tensor Gauss rule wastes almost all
of its points. 

The integral is first split along the diagonal, 
\begin{align*}
\int_{0}^{1}\int_{0}^{1}F(x,y)\,\mathrm{d}y\,\mathrm{d}x & =\int_{0}^{1}\int_{0}^{x}F(x,y)\,\mathrm{d}y\,\mathrm{d}x+\int_{0}^{1}\int_{x}^{1}F(x,y)\,\mathrm{d}y\,\mathrm{d}x,\\
F(x,y) & =\mathrm{e}^{-b_{i}(x-y)^{2}}\phi_{\alpha}(x)\phi_{\beta}(y),
\end{align*}
and each triangle is mapped back to the unit square by 
\[
(x,y)=\bigl(\xi,(1-\eta)\xi\bigr)\quad\text{or}\quad(x,y)=\bigl(1-\xi,(\eta-1)\xi+1\bigr),\qquad(\xi,\eta)\in[0,1]^{2}.
\]
The Jacobian is $\xi$, and $(x-y)^{2}$ becomes $\xi^{2}\eta^{2}$.
Adding both contributions yields 
\[
\int_{0}^{1}\int_{0}^{1}\xi\,\mathrm{e}^{-b_{i}\xi^{2}\eta^{2}}\Bigl[\,\phi_{\alpha}(\xi)\,\phi_{\beta}\bigl((1-\eta)\xi\bigr)+\,\phi_{\alpha}(1-\xi)\,\phi_{\beta}\bigl((\eta-1)\xi+1\bigr)\Bigr]\,\mathrm{d}\eta\,\mathrm{d}\xi.
\]
The peak is now restricted to the edges $\xi=0$ and $\eta=0$. As
already shown, the integral is nonzero only for matching parity of
$\phi_{\alpha}$ and $\phi_{\beta}$. Additionally, if the parity
matches, we have $\phi_{\alpha}(x)=\phi_{\alpha}(1-x)$ for even $\phi_{\alpha}$,
and $\phi_{\alpha}(x)=-\phi_{\alpha}(1-x)$ for odd $\phi_{\alpha}$,
resulting in the simplification of the integral to 
\[
2\int_{0}^{1}\int_{0}^{1}\xi\,\mathrm{e}^{-b_{i}\xi^{2}\eta^{2}}\,\phi_{\alpha}(\xi)\,\phi_{\beta}\bigl((1-\eta)\xi\bigr)\,\mathrm{d}\eta\,\mathrm{d}\xi.
\]

To distribute the peak from the edges over a larger portion of the
square, we stretch those edges algebraically: 
\[
\xi=u^{g_{x}},\qquad\eta=v^{g_{y}},\qquad(u,v)\in[0,1]^{2},
\]
with exponents $g_{x},g_{y}\ge1$. The resulting integral becomes
\[
2\int_{0}^{1}\int_{0}^{1}g_{x}g_{y}\,u^{g_{x}}v^{g_{y}}\exp\bigl(-b_{i}u^{2g_{x}}v^{2g_{y}}\bigr)\Theta_{\alpha\beta}\bigl(u^{g_{x}},v^{g_{y}}\bigr)\,\mathrm{d}v\,\mathrm{d}u,
\]
where $\Theta_{\alpha\beta}$ contains the polynomial factors. Since
we use this formulation to compute the entire matrix (i.e., the same
maximal polynomial degree for both $\alpha$ and $\beta$), we may
restrict ourselves to the same constant for both dimensions, $g=g_{x}=g_{y}$.

\subsubsection{Numerical experiment}

\label{subsec:duffy-numex}

\begin{figure}
\begin{centering}
\subfloat[Basis size 30]{\begin{centering}
\includegraphics{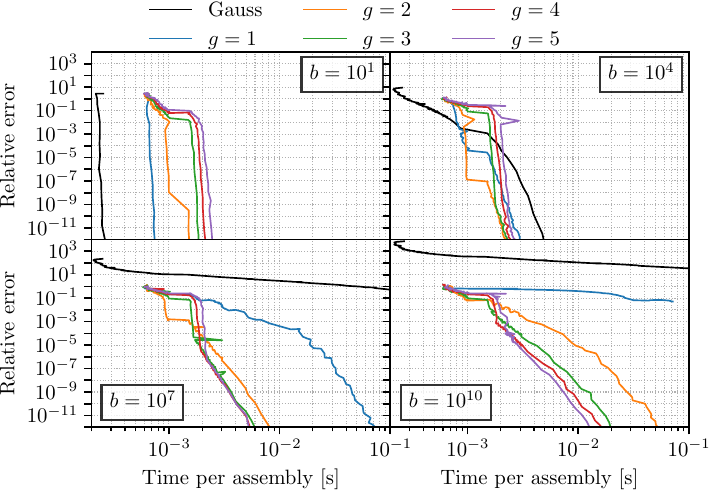} 
\par\end{centering}
}
\par\end{centering}
\begin{centering}
\subfloat[Basis size 60]{\begin{centering}
\includegraphics[clip,viewport=0bp 0bp 341bp 210bp]{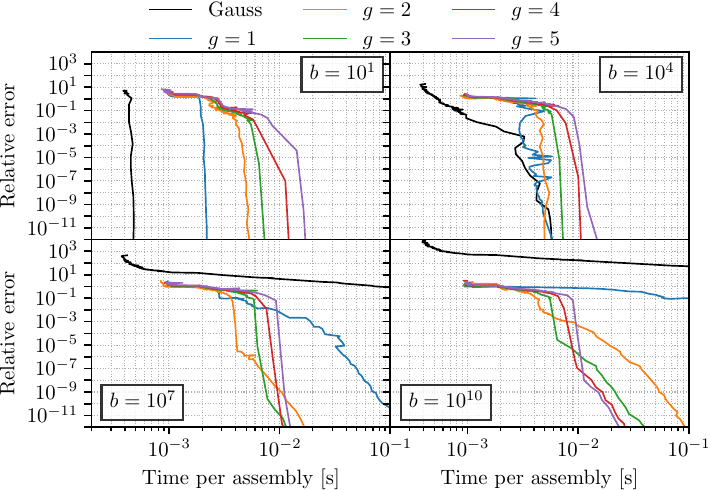} 
\par\end{centering}
}
\par\end{centering}
\centering{}\caption{Efficiency of the proposed Duffy-based quadrature (color; algebraic
exponents $g=1\ldots5$) compared to a standard tensor Gauss rule
(black). Each panel corresponds to a different kernel coefficient
$b_{i}$.}
\label{fig:Comparison-of-efficiency-quadrature} 
\end{figure}

The algebraic stretching introduced above entails additional computational
effort compared to a standard tensor Gauss rule. Each basis polynomial
must now be evaluated at all points of the tensor grid, as its argument
is $\bigl((1-\eta)\xi\bigr)$. This means we transition from a linear
to a quadratic number of evaluation points, resulting in increased
complexity when aggregating contributions. 

To quantify the overall impact, we measured the time required to assemble
a single matrix $\mathbf{A}_{n,i}^{(l)}$ (for 30/60 basis functions,
i.e., up to polynomials of degree 29/59, evaluated separately for
even and odd blocks) on an \textbf{AMD Ryzen 7 7735U} CPU using the
vectorized \texttt{NumPy} and \texttt{opt\_einsum} packages, along
with an optimized BLAS/LAPACK build via the \texttt{blas-openblas}
package. Figure~\ref{fig:Comparison-of-efficiency-quadrature} shows
the resulting runtime versus the achieved relative Frobenius error
of the assembled matrix for 
\[
b_{i}\in\left\{ 10^{1},10^{4},10^{7},10^{10}\right\} ,\qquad g\in\{1,2,3,4,5\},
\]
alongside the reference tensor Gauss rule (black). 

We observe the following: 
\begin{itemize}
\item The plain Gauss rule is slightly faster for moderate values $b_{i}<10^{4}$,
as the kernel remains sufficiently broad. In this setting, the overhead
from the transformation dominates. 
\item For each value of $b_{i}$, there exists an optimal value of $g$.
Within the tested range, this optimal $g$ varies from $2$ to $5$.
In practice, values of $b_{i}$ exceeding $10^{10}$ are uncommon,
except in highly precise squared-exponential approximations. This
is primarily because the stretching induced by high values of $g$
also compresses regions farther from the axes, increasing the frequency
of oscillating basis functions and, consequently, the number of quadrature
points needed for accurate integration. 
\item Comparing the two basis sizes reveals similar overall behavior. For
lower values of $b_{i}$, the reference tensor Gauss rule becomes
marginally more efficient than the Duffy-based quadrature at higher
polynomial degrees. 
\end{itemize}
In summary, while the Duffy transformation introduces a modest overhead
for wide kernels, it becomes essential once $b_{i}\gtrsim10^{4}$.
For the extreme localization characteristic of high-rank squared-exponential
approximations ($b_{i}\ge10^{7}$), the plain Gauss tensor product
fails to achieve useful accuracy at a reasonable cost, whereas the
proposed scheme remains both stable and efficient. 

\section{Numerical experiments -- 2-D sampling }\label{sec:numerical-samples}

With the separable squared-exponential covariance fits from Section~\ref{sec:numerical_examples_square_approx}
and the tensor-product Legendre Galerkin discretisation introduced
in Section~\ref{sec:galerkin}, we computed truncated Karhunen--Loève
expansions on the unit square $[0,1]^{2}$ for all kernels listed
in Subsection~\ref{subsec:Problem-setting-kernels}, along with the
squared-exponential kernel $e^{-d^{2}}$. The parity block strategy
from Section~\ref{subsec:parity-blocks} enables efficient evaluation
of the leading eigenpairs.

\begin{figure}[H]
\centering{}\includegraphics{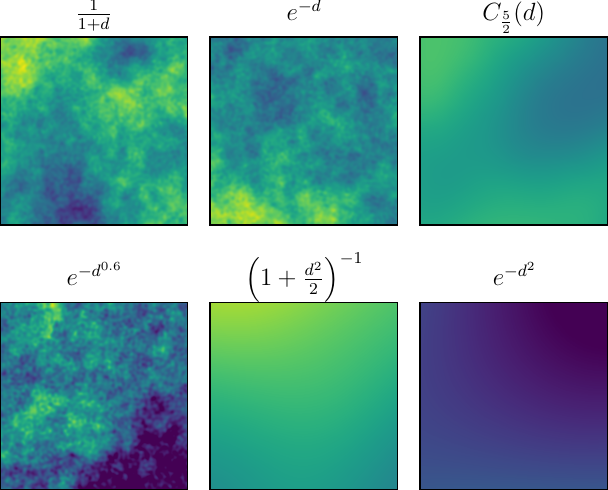} \caption{Random samples of the 2-D Gaussian fields on $[0,1]^{2}$ for the
benchmark kernels and their squared-exponential approximations (with
precision $10^{-6}$); see Table~\ref{tab:ranks}.}
\label{fig:kl-samples} 
\end{figure}

\begin{figure}[H]
\centering{}\includegraphics{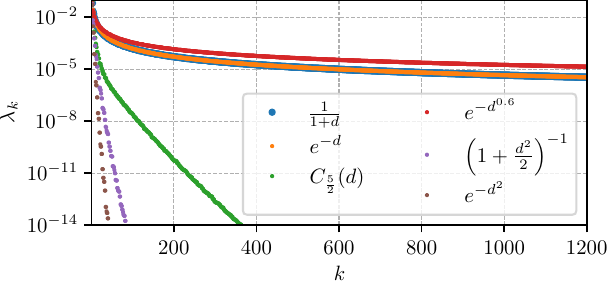} \caption{Eigenvalues of the 2-D KL expansion for the benchmark kernels and
their squared-exponential approximations (with precision $10^{-6}$);
see Table~\ref{tab:ranks}.}
\label{fig:kl-eigenvalues} 
\end{figure}

To ensure accurate approximations of both samples and eigenvalues,
we use the highest-precision squared-exponential fit, i.e., $\varepsilon=10^{-6}$,
with 150 basis functions in each dimension ($n=149$), and compute
the 1500 largest eigenvalues for each parity block. Figure~\ref{fig:kl-samples}
presents representative random-field realisations obtained by drawing
independent standard Gaussian coefficients and summing the resulting
KL series, while Figure~\ref{fig:kl-eigenvalues} shows the numerically
computed eigenvalues. A clear relationship is evident between eigenvalue
decay and the short-range correlation between points. 

\section{Accuracy of the one--dimensional KL approximation }\label{sec:error-1d}

We conclude with a rough a~posteriori error assessment of the KL
discretisation. First, we consider the residual norm of an eigenpair
approximation: 
\[
\mathcal{R}(u,\lambda)\;=\;\left\Vert \int_{0}^{1}C(x,y)\,u(y)\,\mathrm{d}y\;-\;\lambda\,u(x)\right\Vert _{L^{2}(0,1)},
\]
evaluated via numerical quadrature (accounting for possible non-smoothness
of the kernel at $x=y$). Second, we assess the error in reconstructing
the covariance kernel: 
\[
\lVert C-C_{N}\rVert_{L^{2}([0,1]^{2})}\quad\left(C_{N}(x,y)={\textstyle \sum_{j=1}^{N}\lambda_{j}\phi_{j}(x)\phi_{j}(y)}\right),
\]
again computed via numerical quadrature, with care taken near the
potential non-smoothness at $x=y$.

For both types of error estimation, we examine the effects of the
squared-exponential approximation and the polynomial Galerkin approximation.

The numerical study is restricted to $\Omega=[0,1]$, as the computational
effort required to estimate $\mathcal{R}(u,\lambda)$ and $\lVert C-C_{N}\rVert_{L^{2}}$
becomes prohibitive in higher dimensions. However, the qualitative
behaviour is expected to carry over to higher dimensions. 

\subsection{Residual decay of selected eigenpairs}

\label{subsec:residuals} 

Figure~\ref{fig:residual-convergence} shows the dependence of the
residual error $\mathcal{R}(u_{j},\lambda_{j})$ on the polynomial
degree of the Galerkin basis. It is presented for the 1st, 15th, and
30th eigenpairs of the exponential and Matérn\,($\nu=\tfrac{5}{2}$)
kernels, as the tolerance~$\varepsilon$ of the squared-exponential
fit is refined ($10^{-2},10^{-4},10^{-6}$); see Table~\ref{tab:ranks}.
For the dominant mode ($j=1$), the residual decays exponentially
until it reaches the plateau set by~$\varepsilon$. Higher modes
may already be at their lowest achievable error across all polynomial
approximations, as their exact eigenvalues are smaller than~$\varepsilon$;
see Table~\ref{tab:eigvals-small}. 

\begin{figure}[H]
\centering \includegraphics{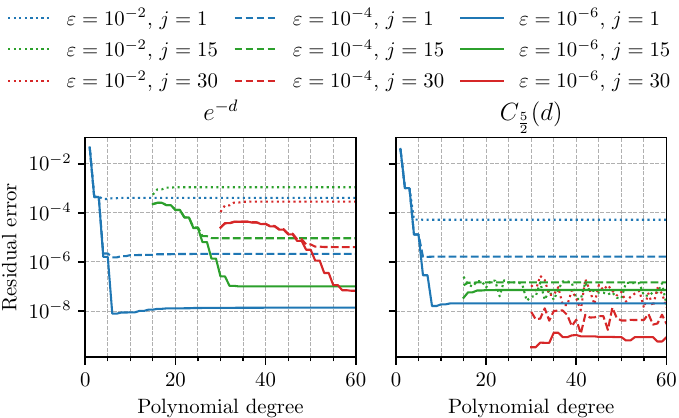} \caption{$L^{2}$ residual $\mathcal{R}(u_{j},\lambda_{j})$ for $j=1,15,30$
under three squared-exponential tolerances $\varepsilon\in\{10^{-2},10^{-4},10^{-6}\}$;
see Table~\ref{tab:ranks}.}
\label{fig:residual-convergence} 
\end{figure}

\begin{table}[H]
\centering \caption{Exact eigenvalues corresponding to Figure~\ref{fig:residual-convergence}.
Values far below the fitting tolerance explain the instant convergence
of the residual.}
\label{tab:eigvals-small} %
\begin{tabular}{lccc}
\toprule 
kernel  & $\lambda_{1}$  & $\lambda_{15}$  & $\lambda_{30}$\tabularnewline
\midrule 
exponential  & $7.388\times10^{-1}$  & $1.031\times10^{-3}$  & $2.409\times10^{-4}$\tabularnewline
Matérn $\nu=\tfrac{5}{2}$  & $8.950\times10^{-1}$  & $3.210\times10^{-9}$  & $5.661\times10^{-16}$\tabularnewline
\bottomrule
\end{tabular}
\end{table}

\subsection{Convergence of the truncated covariance}

\label{subsec:cov-error} 

Figure~\ref{fig:cov-error} reports the error in reconstructing the
covariance, $\lVert C-C_{N}\rVert_{L^{2}}$, for all five benchmark
kernels as the degree of the polynomial Galerkin approximation increases.
Two representative squared-exponential tolerances are used: $\varepsilon=10^{-3}$
and $10^{-6}$ (see Table~\ref{tab:ranks}). Here, we use the full
spectral decomposition---not just the largest eigenvalues, as this
is feasible in the 1D case---to ensure that the estimate of $\lVert C-C_{N}\rVert_{L^{2}}$
remains accurate. The convergence rate closely follows the eigenvalue
decay shown in Figure~\ref{fig:kl-eigenvalues}. Note that the results
in the figure correspond to the 2D case. 

An important observation can be made from the exponential kernel results
in Figure~\ref{fig:cov-error}. For kernels with slowly decaying
eigenvalues, which require more squared-exponential terms to achieve
a given $\varepsilon$, larger values of $\varepsilon$ may still
be acceptable---particularly when only a limited number of eigenvalues
are needed. In such cases, the computed eigenpairs yield a similar
approximation of $C$ for both fine and coarse tolerances. 

\begin{figure}[H]
\centering \includegraphics{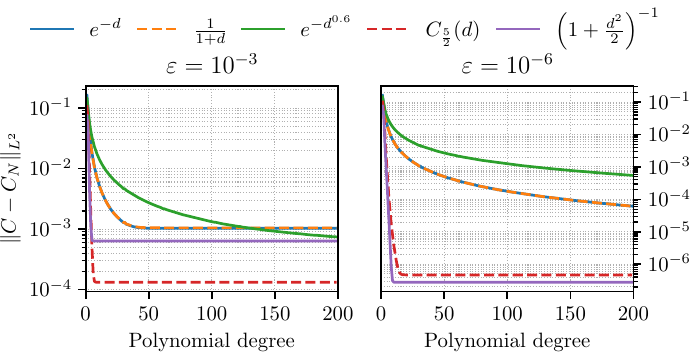} \caption{$\lVert C-C_{N}\rVert_{L^{2}}$ error of the truncated covariance
$C_{N}$ for $\varepsilon=10^{-3}$ (left) and $\varepsilon=10^{-6}$
(right); see Table~\ref{tab:ranks}.}
\label{fig:cov-error} 
\end{figure}

\section{Conclusions}\label{sec:Conclusions}

The numerical approximation of the Karhunen--Loève expansion relies
on orthogonal polynomials at several stages, making it natural to
present both topics together in a single text. Section~\ref{sec:derivation}
demonstrates how the three-term recurrence relation for Legendre polynomials
can be derived without resorting to differential equations or generating
functions, thereby keeping the argument concise, clear, and accessible
to undergraduates.

Sections~\ref{sec:gauss_sum} and \ref{sec:numerical_examples_square_approx}
explain how every positive-definite isotropic covariance function
can be expressed as an integral of squared-exponentials, and how a
finite, non-negative squared-exponential approximation yields a close
representation. Table~\ref{tab:ranks} and Figure~\ref{fig:cov-approx}
illustrate that the approximation error decreases almost geometrically
as the rank increases.

Given a separable kernel, Section~\ref{sec:galerkin} constructs
a Legendre--Galerkin matrix that decomposes into Kronecker products
of small one-dimensional blocks. Due to cancellation between even
and odd Legendre modes, each block further splits, reducing both storage
and computational effort in higher dimensions. The Duffy mapping with
algebraic stretching, employed in Section~\ref{subsec:duffy-quadrature},
preserves quadrature stability even when the exponents in the squared-exponential
approximation become very large; see Figure~\ref{fig:Comparison-of-efficiency-quadrature}.

Figure~\ref{fig:kl-samples} presents two-dimensional samples generated
by the method, while Figure~\ref{fig:kl-eigenvalues} displays the
corresponding eigenvalues. Together, they demonstrate that the numerical
scheme accurately captures both the spatial structure and spectral
decay of the target fields.

All scripts used to generate the tables and figures are available
at \url{https://github.com/Beremi/KL_decomposition}, facilitating
verification and reuse of the results with minimal effort.

\bibliographystyle{siam}
\bibliography{Beres_AOM_2025_submit}

\end{document}